\documentclass[twocolumn,switch]{article}
\usepackage{arxiv}
\usepackage{graphicx}
\usepackage{latexsym}
\usepackage{amsmath}
\usepackage{amssymb}
\usepackage{enumitem}
\usepackage{url}

\usepackage{float}
\usepackage{tikz}
\usepackage{fancyvrb}
\usepackage{listings}
\usepackage{xcolor}
\usepackage{color}
\usepackage{soul}
\usepackage{multirow}
\usepackage{xspace}
\usepackage{multicol}
\usepackage{subfig}
\usepackage{booktabs}
\usepackage{caption}
\usepackage{wrapfig}
\graphicspath{{content/}{image/}{./}}

\providecommand{\mat}[1]{\boldsymbol{\mathrm{#1}}}%
\renewcommand{\vec}[1]{\boldsymbol{\mathrm{#1}}}

\DeclareMathOperator{\argmax}{argmax}
\DeclareMathOperator{\argmin}{argmin}

\newcommand{\mc}{\mathcal{C}}
\providecommand{\mA}{\ensuremath{\mat{A}}}

\providecommand{\vb}{\ensuremath{\vec{b}}}
\providecommand{\vc}{\ensuremath{\vec{c}}}

\providecommand{\vx}{\ensuremath{\vec{x}}}
\providecommand{\vy}{\ensuremath{\vec{y}}}

\newtheorem{theorem}{Theorem}

\newtheorem{observation}[theorem]{Observation}

\newenvironment{proof}{\medskip \par \noindent {\bf Proof}\ }{\hfill
	$\Box$
	\medskip \par}

\usepackage{caption}

\usepackage{mathrsfs}

\newcommand{\xhdr}[1]{\vspace{0.0mm}\noindent{\textbf{#1.}}\hspace{0.5mm}}

\newcommand{\dataheader}[1]{\vspace{-0.2mm}\noindent{\emph{#1}.}\hspace{0.15mm}}

\usepackage{hyperref}
\definecolor{mylinkcolor}{RGB}{0,0,140}
\hypersetup{colorlinks,allcolors=mylinkcolor,citecolor=mylinkcolor}

\usepackage{xcolor}

\author{
\begin{tabular}[t]{c@{\extracolsep{7em}} c@{\extracolsep{7em}} c}
Ilya Amburg & Nate Veldt & Austin R. Benson\\ 
\small Cornell University & \small Cornell University & \small Cornell University\\ 
\small ia244@cornell.edu & \small  nveldt@cornell.edu & \small arb@cs.cornell.edu\\
\end{tabular}
}
\title{\bf Hypergraph Clustering for Finding \\ Diverse and Experienced Groups}

\begin{document}
\twocolumn[ 
  \begin{@twocolumnfalse} 
\maketitle

\begin{abstract}
  When forming a team or group of individuals, we often seek a balance of expertise in a particular
  task while at the same time maintaining diversity of skills within each group.
  Here, we view the problem of finding diverse and experienced groups as clustering in hypergraphs with multiple edge types.
  The input data is a hypergraph with multiple hyperedge types --- representing information about past experiences of groups of individuals --- and the output is groups of nodes.  In contrast to related problems on fair or balanced clustering, we model diversity in terms of variety of past \textit{experience} (instead of, e.g., protected attributes), 
  with a goal of forming groups that have both experience and diversity with respect to participation in edge types.
  In other words, both diversity and experience are measured from the types of the hyperedges.

  Our clustering model is based on a regularized version of an edge-based hypergraph clustering objective,
  and we also show how naive objectives actually have no diversity-experience tradeoff.
  Although our objective function is NP-hard to optimize, we design an efficient 2-approximation algorithm
  and also show how to compute  bounds for the regularization hyperparameter that lead to meaningful diversity-experience tradeoffs. 
  We demonstrate an application of this framework in online review platforms, where the goal is to curate sets of user reviews for a product type. In this context, ``experience'' corresponds to users familiar with the type of product, and ``diversity'' to users that have reviewed related products.
\end{abstract}

  \end{@twocolumnfalse} 
] 

\section{Introduction}
Team formation within social and organizational contexts is ubiquitous in modern society, as success often relies on forming the ``right'' teams. Diversity within these teams, both with respect to socioeconomic attributes and expertise across disciplines,  often leads to synergy, and brings fresh perspectives which facilitate innovation. The study of diverse team formation with respect to expertise has a rich history spanning decades of work in sociology, psychology and business management~\cite{forsyth2018group,homan2008facing,jackson1995diversity,levi2015group}. In this paper, 
we explore algorithmic approaches to diverse team formation, where ``diversity'' corresponds to a group of individuals with a variety of experiences.
In particular, we present a new algorithmic framework for diversity in clustering that focuses on forming groups which are \emph{diverse} and \emph{experienced} in terms of past group interactions. As a motivating example,
consider a diverse team formation task in which the goal is to assign a task to a
group of people who (1) already have some level of experience working together
on the given task, and (2) are diverse in terms of their previous work
experience. As another example, a recommender system may want to display a
diverse yet cohesive set of reviews for a certain class of products. In other
words, the set of displayed reviews should work well together in providing an
accurate overview of the product category. 

Here, we formalize diverse and experienced group formation as a
clustering problem on edge-labeled hypergraphs. In this setup, a hyperedge
represents a set of objects (such as a group of individuals) that have
participated in a group interaction or experience. The hyperedge label encodes
the \emph{type} or \emph{category} of interaction (e.g., a type of team
project). The output is then a labeled clustering of nodes, where cluster labels
are chosen from the set of hyperedge labels. The goal is to form clusters whose
nodes are balanced in terms of \emph{experience} and \emph{diversity}. By
\emph{experience} we mean that a cluster with label $\ell$ should contain nodes
that have previously participated in hyperedges of label $\ell$.
By \emph{diversity}, we mean that clusters should also include nodes
that have participated in other hyperedge types.

Our mathematical framework for diverse and experienced clustering builds on an
existing objective for clustering edge-labeled
hypergraphs~\cite{amburg2020clustering}. This objective encourages cluster
formation in such a way that hyperedges of a certain label tend to be contained
in a cluster with the same label, similar to (chromatic) correlation
clustering~\cite{bansal2004correlation,bonchi2015chromatic}. We add a
regularization term to this objective that encourages clusters to contain nodes that
have participated in many different hyperedge types. This diversity-encouraging
regularization is governed by a tunable parameter $\beta \geq 0$, where $\beta =
0$ corresponds to the original edge-labeled clustering objective. Although the
resulting objective is NP-hard in general, we design a linear programming
algorithm that guarantees a 2-approximation for any choice of $\beta$. 
We show that certain values of this hyperparameter reduce to extremal solutions with
closed-form solutions where just diversity or just experience is encouraged. In
order to guide a meaningful hyperparameter selection, we show how to bound the
region in which non-extremal solutions occur, leveraging linear programming
sensitivity techniques. 

We demonstrate the utility of our framework by applying it to team formation of users posting answers on Stack Overflow, and the task of aggregating a diverse set of reviews for categories of establishments and products on review sites (e.g., Yelp or Amazon). We find that the framework yields meaningfully more diverse clusters at a small cost in terms of the unregularized objective, and the approximation algorithms we develop produce solutions within a factor of no more than 1.3 of optimality empirically. A second set of experiments examines the effect of iteratively applying the diversity-regularized objective while keeping track of the experience history of every individual. 
We observe in this synthetic setup that regularization greatly influences team formation dynamics over time,
with more frequent role swapping as regularization strength increases.

\subsection{Related work}
Our work on diversity in clustering is partly related to the recent research on algorithmic fairness and fair clustering, which we briefly review. 
These results are based ideas that machine learning 
algorithms may make decisions that are biased or unfair towards a subset of a population~\cite{barocas2016big,chouldechova2017fair,corbett2018measure}.
There are now a variety of algorithmic fairness techniques to combat this
issue~\cite{corbett2017algorithmic,kleinberg2018algorithmic,mehrabi2019survey}.
For clustering problems, fairness is typically formulated in terms of protected attributes on 
data points --- a cluster is ``fair'' if it exhibits a proper balance between nodes from different
protected classes, and the goal is to 
optimize classical clustering objectives while also adhering to constraints
ensuring a balance on the protected attributes~\cite{ahmadi2020fair,ahmadian2019correlation,ahmadian2019clustering,bera2019fair,chen2019proportionally,chierichetti2017fair,davidson2019making,kleindessner2019guarantees}.
 These approaches are similar to our notion of diverse clustering, since in both cases, the resulting clusters are more heterogeneous with respect to node attributes. While the primary attribute addressed in fair clustering is the protected status of a data point, 
 in our case it is the ``experience'' of that point, which is derived from an edge-labeled hypergraph. 
 In this sense, we have similar algorithmic goals, but we present our approach as a general framework for finding clusters
 of nodes with experience and diversity, as there are many reasons why one may wish to organize a dataset into clusters that are diverse with respect to past experience.

Furthermore, there exists some literature on algorithmic diverse team formation~\cite{aziz2019rule,celis2018multiwinner,kleinberg2018team,machado2019fair,stratigi2020fair,zehlike2017fa}. 
However, this research has thus far largely focused on frameworks for forming teams which exhibit diversity with respect to inherent node-level attributes, 
without an emphasis on diversity of expertise, making our work the first to explicitly address this issue.

Our framework also facilitates a novel take on diversity within recommender systems. A potential application which we study in Section~\ref{sec:experiment} is to select expert, yet diverse sets of reviews for different product categories. This differs from existing recommender systems paradigms on two fronts: First, the literature focuses almost exclusively on user-centric recommendations, with content tailored to an individual's preferences; for us, a set of reviews is curated for a {\it category} of products that allows any user to glean both expert and diverse opinions regarding it. Further, the classic recommender systems literature defines diversity for a set of objects based on dissimilarity derived from pairwise relations among people and objects~\cite{castells2015novelty,kunaver2017diversity,lu2012recommender,vargas2011rank,zhou2010solving}. While some set-level proxies for diverse recommendations exist~\cite{chen2018fast,santos2010exploiting}, they do not deal explicitly with higher-order interactions among objects. Our work, on the other hand, encourages diversity in recommendations through an objective that captures higher-order information about relations between subsets of objects, providing a native framework for hypergraph-based diverse recommendations.

\section{Forming Clusters Based on Diversity and Experience}
We first introduce notation for edge-labeled clustering.
After, we analyze an approach that seems natural for clustering based on experience and
diversity, but leads only to trivial solutions. We then show how a regularized
version of a previous hypergraph clustering framework leads to a more meaningful objective,
which will be the focus of the remainder of the paper.

\xhdr{Notation} Let $G = (V,E,L, \ell)$ be a hypergraph with labeled edges,
where $V$ is the set of nodes, $E$ is the set of (hyper)edges, $L$ is a set of
edge labels, and $\ell\colon E \rightarrow L$ maps edges to labels, where
$L = \{1, \ldots, k\}$ and $k$ is the number of labels.
Furthermore, let $E_c \subseteq E$ be the edges with label $c$, and $r$ the maximum hyperedge
size. Following graph-theoretic terminology, we often refer to elements in $L$
as ``colors''; in data, $L$ represents categories or types.
For any node $v \in V$, let $d_v^c$ be the number of hyperedges of
color $c$ in which node $v$ appears. We refer to $d_v^c$ as the color degree of
$v$ for color $c$.

Given this input, an algorithm will output a color for each node. Equivalently,
this is a clustering $\mathcal{C}$, where each node is assigned to exactly one
cluster, and there is exactly one cluster for each color in $L$.
We use $\mathcal{C}(i)$ to denote the nodes assigned to
color $i \in L$. We wish to find a clustering that promotes both diversity
(clusters have nodes from a range of colored hyperedges), and experience (a
cluster with color $c$ contains nodes that have experience participating in
hyperedges of color $c$).

\subsection{A flawed but illustrative first approach}
We start with an illustrative (but ultimately flawed) clustering objective
whose characteristics will be useful in the rest of the paper.
For this, we first define \emph{diversity and experience scores} for a color $i$, denoted $D(i)$ and $E(i)$, as follows:
\begin{equation}
\label{eq:dande}
D(i)=\sum_{v\in \mc(i),c\neq i}d_v^c\,, \hspace{.5cm} E(i)=\sum_{v\in \mc(i)}d_v^i.
\end{equation}
In words, $D(i)$ measures how much nodes in cluster $i$ have
participated in hyperedges that are \emph{not} color $i$, and $E(i)$
measures how much nodes in cluster $i$ have participated in
hyperedges of color $i$. A seemingly natural but ultimately naive objective for
balancing experience and diversity is:
\begin{equation}
\label{eq:naiveobj}
\max_\mc \sum_{i\in L}[E(i)+\beta D(i)].
\end{equation}
The regularization parameter $\beta$ determines the relative importance of the
diversity and experience scores. It turns out that the optimal solutions to this
objective are overly-simplistic, with a phase transition at $\beta=1$.
We define two simple types of clusterings as follows:
\begin{itemize} 
\item \emph{Majority vote clustering}: Node $v$ is placed in cluster $\mc(i)$ where $i \in \argmax_{c \in L} d_v^c$,
  i.e., node $v$ is placed in a cluster for which it has the most experience.
\item \emph{Minority vote clustering}: Node $v$ is placed in cluster $\mc(i)$ where $i \in \argmin_{c \in L} d_v^c$,
  i.e., node $v$ is placed in a cluster for which it has the least experience.
\end{itemize}

The following theorem explains why \eqref{eq:dande} does not provide a meaningful tradeoff between diversity and experience.
\begin{theorem}
  A majority vote clustering optimizes~\eqref{eq:naiveobj} for all $\beta > 1$,
  and a minority vote clustering optimizes the same objective for
  all $\beta < 1$. Both are optimal when $\beta = 1$.
\end{theorem}
\begin{proof}
  For node $i$, assume without loss of generality that the colors $1, 2, \hdots, k$ are ordered so that $d_i^1\geq\dots\geq d_i^k$.
  Clustering $i$ with cluster 1 gives a contribution of $d_i^1+\beta\sum_{j=2}^{k}d_i^j$ to the objective, while clustering it with color $c\neq i$ gives contribution $d_i^c+\beta\sum_{j\neq c}d_i^j$.
  Because $d_i^1\geq\dots\geq d_i^k$, the first contribution is greater than or equal to the second if and only if $\beta\leq 1$.
  Hence, majority vote is optimal when $\beta \geq 1$.
  A similar argument proves optimality for minority vote when $\beta \leq 1$.
\end{proof}
Objective~\eqref{eq:naiveobj} is easy to analyze, but has optimal points that do
not provide a balance between diversity and experience. This occurs because a
clustering will maximize the total diversity $\sum_{c \in L} D(c)$ if and only
if it minimizes the total experience $\sum_{c \in L} E(c)$, as these terms sum to a constant.
We formalize this in the following observation.
\begin{observation}\label{obs:constant}
$\sum_{c \in L} [E(c)+D(c)]$ is a constant independent of the clustering $\mathcal{C}$.
\end{observation}
We will use this observation when developing our clustering framework in the next section.

\subsection{Diversity-regularized categorical edge clustering}
We now turn to a more sophisticated approach: a regularized version of the
\emph{categorical edge clustering} objective~\cite{amburg2020clustering}. For
a clustering $\mc$, the objective accumulates a penalty of
1 for each hyperedge of color $c$ that is not completely contained in the cluster
$\mathcal{C}(c)$. More formally, the objective is:
\begin{equation}
\label{obj:cec}
\min_\mc \sum_{c \in L} \sum_{e \in E_c} x_e,
\end{equation}
where $x_e$ is an indicator variable equal to 1 if hyperedge $e \in E_c$ is
\emph{not} contained in cluster $\mathcal{C}(c)$, but is zero otherwise. This
penalty encourages \emph{entire} hyperedges to be contained
inside clusters of the corresponding color.
For our context, this objective can be interpreted as promoting \emph{group} experience in cluster formation: if a
group of people have participated together in task $c$, this is an indication
they could work well together on task $c$ in the future. However, we want to
avoid the scenario where groups of people endlessly work on the same type of task
without the benefiting from the perspective of others with different
experiences. Therefore, we regularize objective~\eqref{obj:cec} with a
penalty term $\beta \sum_{c \in L} E(c)$. Since $\sum_{c \in L} [E(c)+D(c)]$ is a constant (Observation~\ref{obs:constant}),
this regularization encourages
higher diversity scores $D(c)$ for each cluster $\mc(c)$. 

While the ``all-or-nothing'' penalty in \eqref{obj:cec} may seem restrictive at first, it is a natural choice for our objective function for several reasons.
First, we are building on recent research showing applications of Objective \eqref{obj:cec} on datasets similar to ours,
namely edge-labeled hypergraphs~\cite{amburg2020clustering},
and this type of penalty is a standard in hypergraph partitioning~\cite{benson2016higher,hadley1995,ihler1993modeling,li2017inhomogeneous,lawler1973cutsets,tsourakakis2017scalable}.
Second, if we consider an alternative penalty which incurs a cost of one for every node that is split away
from the color of the hyperedge, this reduces to the ``flawed first approach'' in the previous section, where there is no diversity-experience tradeoff.
Developing algorithms that can optimize more complicated alternative hyperedge cut penalties is an active area of research~\cite{li2017inhomogeneous,veldt2020hypergraph}.
Translating these ideas to our setting constitutes an interesting open direction for future work, but comes with numerous challenges and is not the focus of the current manuscript. 
Finally, our experimental results indicate that even without considering generalized hyperedge cut penalties, our regularized objective produces meaningfully diverse clusters on real-world and synthetic data.


We now formalize our objective function, which we call  
\emph{diversity-regularized categorical edge clustering} (DRCEC),
which will be the focus for the remainder of the paper.
We state the objective as an integer linear program (ILP):
\begin{equation}
\label{eq:ccilp}
\begin{array}{lll}
\min & \sum_{c \in L} \sum_{e \in E_c} x_e +\beta\sum_{v\in V}\sum_{c\in L}d_v^c(1-x_v^c) &\\[1mm]
\text{s.t.} & \text{for all $v \in V$:}\; \sum_{c = 1}^k x_v^c = k -1, &\\
&\text{for all $c \in L$, $e \in E_c$:}\; x_v^c \leq x_e \text{ for all $v \in e$}; \\
& \text{for all $c \in L$, $v \in V$, $e \in E$:}\; x_v^c, x_e \in \{0,1\}.
\end{array}
\end{equation}
The binary variable $x_v^c$ equals $1$ if node $v$ is not assigned label $c$,
and is 0 otherwise. The first constraint guarantees every node is assigned to
exactly one color, while the second constraint guarantees that if a single node
$v \in e$ is not assigned to the cluster of the color of $e$, then $x_e = 1$.

\xhdr{A polynomial-time 2-approximation algorithm} 
Optimizing the case of $\beta = 0$ is NP-hard~\cite{amburg2020clustering}, so DRCEC is also NP-hard.
Although the optimal solution to~\eqref{eq:ccilp} may vary with $\beta$,
we develop a simple algorithm based on solving an LP relaxation of the ILP that rounds to a 2-approximation for every value of $\beta$.
Our LP relaxation of the ILP in \eqref{eq:ccilp} replaces the binary constraints $x_v^c, x_e \in \{0,1\}$ with linear constraints $x_v^c, x_e \in [0,1]$.
The LP can be solved in polynomial time, and the objective score is a lower bound on the optimal solution score to the NP-hard ILP.
The values of $x_v^c$ can then be \emph{rounded} into integer solutions to produce a clustering that is within a bounded factor of the LP lower bound, and therefore within a bounded factor of optimality.
Our algorithm is simply stated:

\vspace{0.3cm}
$\begin{array}{l}
    \textbf{\underline{Algorithm 1}} \vspace{0.5mm} \\
    
	\text{1. Solve the LP relaxation of the ILP in \eqref{eq:ccilp}.} \\
	\text{2. For each $v \in V$, assign $v$ to any $c \in \argmin_{j}x_v^j$.}
\end{array}$
\vspace{0.3cm}

The following theorem shows that the LP relaxation is a 2-approximation for the ILP formulation
of the DRCEC objective.
\begin{theorem}
  For any $\beta \geq 0$, Algorithm 1 returns a 2-approximation for objective~\eqref{eq:ccilp}.
\end{theorem}
\begin{proof}
  Let the relaxed solution be $\{x_e^*,x_v^{*c}\}_{e\in E,v\in V,c\in L}$ and the corresponding rounded solution $\{x_e,x_v^{c}\}_{e\in E,v\in V,c\in L}$.
  Let $y_v^c=1-x_v^c$ and $y_v^{*c}=1-x_v^{*c}$.
  Our objective evaluated at the relaxed and rounded solutions respectively are
\begin{align*}
   S^*&=\sum_{e}x_e^*+\beta\sum_{v\in V}\sum_{c\in L}d_v^cy_v^{*c} \text{ and } \\
   S &=\sum_{e}x_e+\beta\sum_{v\in V}\sum_{c\in L}d_v^cy_v^{c}.
\end{align*}
We will show that $S\leq 2S^*$ by comparing the first and second terms of $S$
and $S^*$ respectively. The first constraint in~\eqref{eq:ccilp}
ensures that $x_v^c < 1/2$ for at most a single color $c$.
Thus, for every edge $e$ with $x_e = 1$, $x_v^{*c}\geq1/2$ for some $v \in e$.
In turn, $x_e^*\geq1/2$, so $x_e\leq2x_e^*$.
If $x_e = 0$, then $x_e\leq2x_e^*$ holds trivially.
Thus, $\sum_ex_e\leq 2\sum_ex_e^*$.
Similarly, since $x_v^c=1$ ($y_v^c=0$) if and only if $x_v^{*c}\geq 1/2$ ($y_v^{c*}\leq 1/2$), and
$x_v^c=0$ otherwise, it follows that $y_v^c\leq2y_v^{*c}$. Thus,
$\sum_{v\in V}\sum_{c\in L}d_v^cy_v^{c}\leq 2 \sum_{v\in V}\sum_{c\in L}d_v^cy_v^{*c}$.
\end{proof}

\subsection{Extremal LP and ILP solutions at large enough values of $\beta$}
In general, Objective~\eqref{eq:ccilp} provides a meaningful way to balance group experience (the first term) and diversity (the regularization, via Observation~\ref{obs:constant}).
However, when $\beta \rightarrow \infty$, the objective corresponds to simply minimizing experience, (i.e., maximizing diversity),
which is solved via the aforementioned minority vote assignment.
We formally show that the optimal integral solution~\eqref{eq:ccilp}, as well as the relaxed LP solution under certain conditions,
transitions from standard behavior to extremal behavior (specifically, the minority vote assignment) when $\beta$ becomes larger than the maximum degree in the hypergraph.
In Section~\ref{sec:hyperparameter}, we show how to bound these transition points numerically,
so that we can determine values of the regularization parameter $\beta$ that produce meaningful solutions.

We first consider a simple bound on $\beta$ above which minority vote is
optimal. Let $d_\mathit{max}$ be the largest number of edges any node
participates in.
\begin{theorem}
  \label{thm:betalarge}
  For every $\beta > d_\mathit{max}$, a minority vote assignment optimizes \eqref{eq:ccilp}.
\end{theorem}
\begin{proof}
  Let binary variables $\{x_e, x_v^c\}$ encode a clustering for~\eqref{eq:ccilp} that is not a minority vote solution.
  This means that there exists at least one node $v$ such that $x_v^c = 0$ for some color $c \notin \argmin_{i \in L} d_v^i$.
  If we move node $v$ from cluster $c$ to some cluster $m \in \argmin_{i \in L} d_v^i$,
  then the regularization term in the objective would decrease (i.e., improve) by $\beta (d_v^c - d_v^{m}) \geq \beta > d_\mathit{max}$, since degrees are integer-valued and $d_v^c > d_v^{m}$.
  Meanwhile, the first term in the objective would increase (i.e., become worse) by at most $\sum_{e: v \in e} x_e = d_\mathit{max} < \beta$.
  Therefore, an arbitrary clustering that is not a minority vote solution cannot be optimal when $\beta > d_{\mathit{max}}$.
\end{proof}

A slight variant of this result also holds for the LP relaxation.
For a node $v \in V$, let $\mathcal{M}_v \subset L$ be the set of minority vote clusters for $v$, i.e., $\mathcal{M}_v = \argmin_{c \in L} d_v^c$
(treating $\argmin$ as a set).
The next theorem says that for $\beta > d_\mathit{max}$, the LP places all ``weight'' for $v$ on its minority vote clusters.
We consider this to be a \emph{relaxed minority vote LP solution}, and Algorithm 1 will round the LP relaxation to a minority vote clustering.
\begin{theorem}
	\label{thm:betalargelp}
        For every $\beta > d_\mathit{max}$, an optimal solution to the LP relaxation of~\eqref{eq:ccilp} will satisfy $\sum_{c \in \mathcal{M}_v} (1-x_v^c) = 1$ for every $v \in V$.
        Consequently, the rounded solution from Algorithm 1 is a minority vote clustering.
\end{theorem}
\begin{proof}
  Let $\{x_e, x_v^c\}$ encode an arbitrary solution to the \emph{LP relaxation} of~\eqref{eq:ccilp}, and assume specifically that it is \emph{not} a minority vote solution.
  For every $v \in V$ and $c \in L$, let $y_v^c = 1-x_v^c$. The $y_v^c$ indicate the ``weight'' of $v$ placed on cluster $c$, with $\sum_{c \in L} y_v^c = 1$.
  Since $\{x_e, x_v^c\}$ is not a minority vote solution, there exists some $v \in V$ and $j \notin \mathcal{M}_v$ such that $y_v^j = \varepsilon > 0$.

  We will show that as long as $\beta > d_\mathit{max}$, we could obtain a strictly better solution by moving this weight of $\varepsilon$ from cluster $j$ to a cluster in $\mathcal{M}_v$. 
  Choose an arbitrary $m \in \mathcal{M}_v$, and define a new set of variables $\hat{y}_v^j = 0$, $\hat{y}_v^m = y_v^m + \varepsilon$, and $\hat{y}_v^i = y_v^i$ for all other $i \notin \{m,j\}$. Define $\hat{x}_v^c = 1-\hat{y}_v^c$ for all $c \in L$. For any $u \in V$, $u \neq v$, we keep variables the same: $\hat{y}_u^c = y_u^c$ for all $c \in L$. Set edge variables $\hat{x}_e$ to minimize the LP objective subject to the $\hat{y}_c$ variables, i.e., for $c \in L$ and every $e \in E_c$, let	$\hat{x}_e = \max_{u \in e} \hat{x}_u^c$.

  Our new set of variables simply takes the $\varepsilon$ weight from cluster $j$ and moves it to $m \in \mathcal{M}_v$.
  This improves the regularization term in the objective by at least $\beta \varepsilon$:
	\begin{align*}
	 \beta \sum_{c \in L} d_v^c [y_v^c - \hat{y}_v^c]  &= \beta d_v^m(y_v^m - \hat{y}_v^m) + \beta d_v^j (y_v^j - \hat{y}_v^j)\\
	&= -\beta d_v^m \varepsilon + \beta d_v^j \varepsilon = \beta \varepsilon (d_v^j - d_v^m) \\
	&\geq \beta \varepsilon\,.
	\end{align*}
	Next, we will show that the first part of the objective increases by at most $\varepsilon d_\mathit{max}$.
        To see this, note that for $e \in E_j$ with $v \in e$, $\hat{x}_e \geq 1- \hat{y}_v^j = 1 \implies \hat{x}_e = 1$ and $x_e \geq 1-y_v^j = 1- \varepsilon$.
        Therefore, for $e \in E_j$, $v \in e$, we know that $\hat{x}_e - x_e = 1 - x_e \leq 1 - (1-\varepsilon) = \varepsilon$.
        For $e \in E_m$ with $v \in e$ we know $\hat{x}_e - x_e \leq 0$, since $\hat{x}_e = \max_{u \in e} (1- \hat{y}_u^m)$ and $x_e = \max_{u \in e} (1- y_u^m)$, but the only difference between $y_u^m$ and $\hat{y}_u^m$ variables is that $\hat{y}_v^m = y_v^m + \varepsilon \implies (1- \hat{y}_v^m) < (1-y_v^m)$.
        For all other edge sets $E_c$ with $c \notin \{m, j\}$, $\hat{x}_e = x_e$. Therefore, $\sum_{e: v \in e} [\hat{x}_e -x_e] \leq \varepsilon d_{\mathit{max}}$.
	In other words, as long as $\beta > \varepsilon$, we can strictly improve the objective by moving around a positive weight $y_v^j = \varepsilon$ from a non-minority vote cluster $j \notin \mathcal{M}_v$ to some minority vote cluster $m \in \mathcal{M}_v$.
        Therefore, at optimality, for every $v \in V$, $\sum_{c \in \mathcal{M}_v} y_v^c = 1$.	
\end{proof}

Theorem~\ref{thm:betalargelp} implies that if there is a unique minority vote
clustering (i.e., each node has one color degree strictly lower than all
others), then this clustering is optimal for both the original objective and the
LP relaxation whenever $\beta > d_\mathit{max}$.  Whether or not the the optimal
solution to the LP relaxation is the same as the ILP one, the rounded solution
still corresponds to some minority vote clustering that does not meaningfully
balance diversity and experience. The bound $\beta > d_\mathit{max}$ is loose
in practice; our numerical experiments show that the transition occurs for
smaller $\beta$. In the next section, we use techniques in LP sensitivity
analysis that allow us to better bound the phase transition computationally for
a given labelled hypergraph. For this type of analysis, we still need the loose
bound as a starting point.



\section{Bounding Hyperparameters that Yield Extremal Solutions}
\label{sec:hyperparameter}
In order to find a meaningful balance between
experience and diversity, we would like to first find the \emph{smallest}
value of $\beta$, call it $\beta^*$, for which $\beta > \beta^*$ yields a minority vote clustering.
After, we could consider the hyperparameter regime $\beta < \beta^*$.
Given that the objective is NP-hard in general, computing $\beta^*$ exactly may not be
feasible. However, we will show that, for a given edge-labeld hypergraph, we can \emph{compute exactly} the
minimum value $\hat{\beta}$ for which a \emph{relaxed} minority vote solution is
no longer optimal for the LP-relaxation of our objective. This has several
useful implications. First, when the minority vote clustering is unique,
Theorem~\ref{thm:betalargelp} says that this clustering is also optimal
for the ILP for large enough $\beta$. Even when the minority vote clustering is not unique, it may still be
the case that an integral
minority vote solution will still be optimal for the LP relaxation for large
enough $\beta$; indeed, we observe this in experiments with real data later in the paper.
In these cases, we know that $\beta^* \leq \hat{\beta}$, which
allows us to rule out a wide range of parameters leading to solutions that
effectively ignore the \emph{experience} part of our objective. Still, even in cases
where an integral minority vote solution is never optimal for the LP relaxation,
computing $\hat{\beta}$ lets us avoid parameter regimes where Algorithm 1
does not return a minority vote clustering.

Our approach for computing $\hat{\beta}$ is based on techniques for
bounding the optimal parameter regime for a relaxed solution to a clustering
objective~\cite{gan2020graph,nowozin2009solution}. We have adapted these results
for our diversity-regularized clustering objective. This technique can
also be used to compute the largest value of $\beta$ for which LP relaxation of our
objective will coincide with the relaxed solution when $\beta = 0$ (i.e., the unregularized objective),
though we focus on computing $\hat{\beta}$ here.

The LP relaxation of our regularized objective can be written abstractly in the following form
\begin{align}
\label{lp}
\min_{\vx} \,\,  \vc_e^T \vx  + \beta \vc_d^T \vx \,\, \text{ s.t. } \mA \vx \geq \vb,  \vx \geq 0, 
\end{align}
where $\vx$ stores variables $\{x_e, x_v^c\}$, $\mA \vx \geq \vb$ encodes constraints given by the LP relaxation of~\eqref{eq:ccilp}, and $\vc_e, \vc_d$ denote vectors corresponding to the experience and diversity terms in our objective, respectively. Written in this format, we see that this LP-relaxation is a parametric linear program in terms of $\beta$. Standard results on parametric linear programming~\cite{adler1992geometric} guarantee that any solution to~\eqref{lp} for a fixed value of $\beta$ will in fact be optimal for a range of values $[\beta_\ell, \beta_u]$ containing $\beta$. 
The optimal solutions to~\eqref{lp} as a function of $\beta$ correspond to a
piecewise linear, concave, increasing curve, where each linear piece corresponds to a range
of $\beta$ values for which the same feasible LP solution is optimal.

We begin by solving this LP for some $\beta_0 > d_\mathit{max}$, which is guaranteed to produce a solution vector $\vx_{0}$ that is at least a
relaxed form of minority vote (Theorem~\ref{thm:betalargelp}) that would round to a minority vote clustering via Algorithm 1. Our goal is to find the largest value $\hat{\beta}$ for which $\vx_{0}$ no longer optimally solves~\eqref{lp}.
To do so, define $\vc^T = \vc_e^T  + \beta \vc_d^T$ so that we can re-write objective~\eqref{lp} with $\beta = \beta_0$ as
\begin{align}
\label{primal}
\min_{\vx} \,\,  \vc^T \vx \,\, \text{ s.t. } \mA \vx \geq \vb , \vx \geq 0.
\end{align}
 
Finding $\hat{\beta}$ amounts to determining how long the minority vote solution is ``stable'' as the optimal solution to~\eqref{primal}.
Consider a perturbation of~\eqref{primal},
\begin{align}
\label{p-pert}
\min_{\vx} \,\, & \vc(\theta)^T \vx =  \vc^T \vx - \theta \vc_d^T \vx,  \,\, \text{ s.t. } \mA \vx \geq \vb,  \vx \geq 0,
\end{align}
where  $\theta = \beta_0 - \beta$ for some $\beta< \beta_0$, so that~\eqref{p-pert} corresponds to our clustering objective with the new parameter $\beta$.
Since $\vx_0$ is optimal for~\eqref{primal}, it is optimal for~\eqref{p-pert} when $\theta = 0$.
Solving the following linear program provides the range $\theta \in [0, \theta^+]$ for which $\vx_0$ is still optimal for~\eqref{p-pert}:
%
\begin{equation}
\label{aux-primal}
\begin{array}{lll}
\max_{\vy, \theta} & \theta \\[1mm]
\text{s.t.} & \mA^T \vy \leq \vc - \theta \vc_d,\;  \vb^T\vy = \vc^T \vx_0 - \theta \vc_d^T \vx_0.\;
\end{array}
\end{equation}

Let $(\vy^*, \theta^*)$ be the optimal solution to~\eqref{aux-primal}.
The constraints in this LP are designed so that $(\vx_0, \vy^*)$ satisfy primal-dual optimality conditions for the perturbed linear program~\eqref{p-pert} and its dual,
and the objective function simply indicates that we want to find the maximum value of $\theta$ such that these conditions hold.
Thus, $\theta^* = \theta^+$, and $\beta = \beta_0 - \theta^+$ will be the smallest parameter value such that $\vx_0$ is optimal for the LP relaxation of our objective.

Finally, after entering a parameter regime where $\vx_0$ is no longer optimal, the objective function strictly decreases.
Again, by Theorem~\ref{thm:betalargelp}, for large enough $\beta$, the relaxed LP solution is a (relaxed) minority vote one.
Since we find the minimizer of the LP, the solution is the (relaxed) minority vote solution with the smallest objective.
Thus, moving to the new parameter regime will no longer correspond to minority vote, either in the LP relaxation
or in the rounding from Algorithm 1.


\section{Numerical Experiments}\label{sec:experiment}
\begin{table}[t]
  \caption{Summary statistics of datasets. The computed $\hat{\beta}$ bounds
    using the tools in Section~\ref{sec:hyperparameter} are much smaller than
    the $d_{\max}$ bound in Theorem~\ref{thm:betalargelp}.}
  \label{tab:allexp1}
  \centering
  \scalebox{0.80}
  {
  \begin{tabular}{l     l l l l l l}
    \toprule
    \emph{Dataset}              & $\lvert V \rvert$ & $ \lvert E \rvert$ & $L$ & $d_{max}$ & $\hat{\beta}$ \\
    \midrule
    music-blues-reviews         & 1106              & 694                & 7                   & 127       & 0.50          \\
    madison-restaurants-reviews & 565               & 601                & 9                   & 59        & 0.42          \\
    vegas-bars-reviews          & 1234              & 1194               & 15                  & 147       & 0.50          \\
    algebra-questions           & 423               & 1268               & 32                  & 375       & 0.50          \\
    geometry-questions          & 580               & 1193               & 25                  & 260       & 0.50          \\
    \bottomrule
  \end{tabular}
  }
\end{table} 
Here we present two sets of experiments on real-world data to illustrate our theory and methods.
The first involves the diverse clustering objective as stated above
to measure the quality of the LP relaxation and our bounds on $\hat{\beta}$, and
we find that regularization costs little in the objective while substantially improving
the diversity score.
As part of these experiments, we take a closer look at two datasets to investigate
the diversity-experience tradeoff in more depth.
The second set of experiments studies what happens if we apply this clustering iteratively,
where an individual's experience changes over time based on the assignment algorithm;
here, we see a clear effect of the regularization on team dynamics over time.


\xhdr{Datasets}
We first briefly describe the datasets that we use, which come
from online user reviews sites and the MathOverflow question-and-answer site.
In all cases, the nodes in the hypergraphs correspond to users on the given site.
Hyperedges correspond to groups of users that post reviews or answer questions
in a certain time period.
Table~\ref{tab:allexp1} contains summary statistics.

\dataheader{music-blues-reviews}
This dataset is derived from a crawl of Amazon product reviews that contains metadata
on the reviewer identity, product category, and time of reviews~\cite{ni2019justifying}.
We consider all reviews on products that include the tag ``regional blues,''
which is a tag for a subset of vinyl music. 
We partition the reviews into month-long segments.
For each time segment, we create hyperedges of all users who posted
a review for a product with a given sub-tag of the regional blues tag
(e.g., Chicago Blues and Memphis Blues).
The hyperedge category corresponds to the sub-tag.

\dataheader{madison-restuarants-reviews, vegas-bars-reviews}
These two datasets are derived from reviews of establishments on Yelp~\cite{yelp}
for restaurants in Madison, WI and bars in Las Vegas, NV.
We perform the same time segmentation as the music-blues-reviews
dataset, creating hyperedges of groups of users who reviewed
an establishment with a particular sub-tag (e.g., Thai restaurant for Madison,
or cocktail bar for Las Vegas) in a given time segment.
The sub-tag is the hyperedge category.

\dataheader{algebra-questions, geometry-questions}
These two datasets are derived from users answering questions on \texttt{mathoverflow.net}
that contain the tag ``algebra'' or ``geometry''.
We use the same time segmentation and hyperedge construction as
for the reviews dataset, and the sub-tags are given by
all tags matching the regular expressions \texttt{*algebra*} or \texttt{*geometry*} 
(e.g., lie-algebras or hyperbolic-geometry).

\begin{figure*}[t]
    \centering
    \includegraphics[width=0.75\columnwidth]{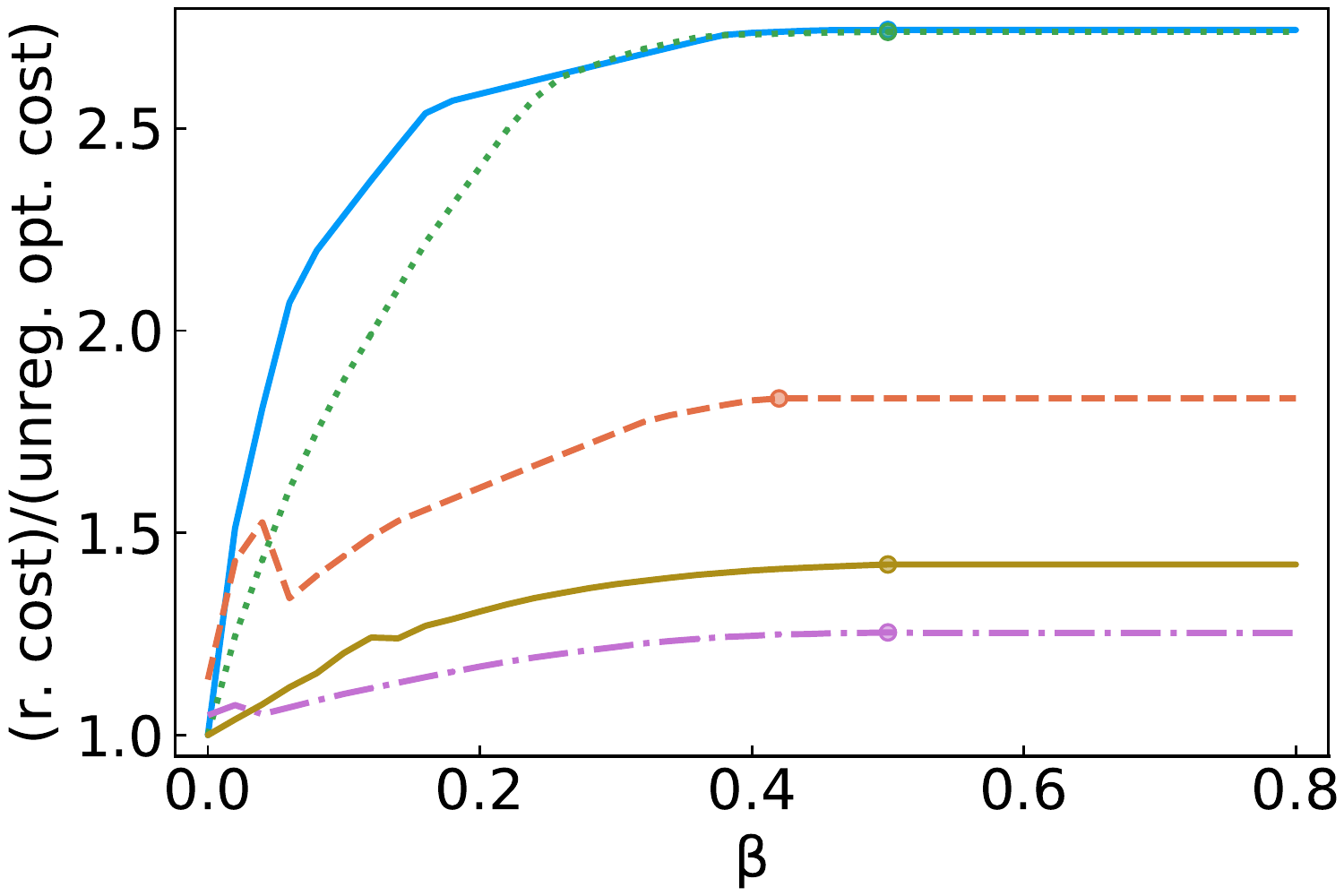}\hspace{1.5cm}
    \includegraphics[width=0.75\columnwidth]{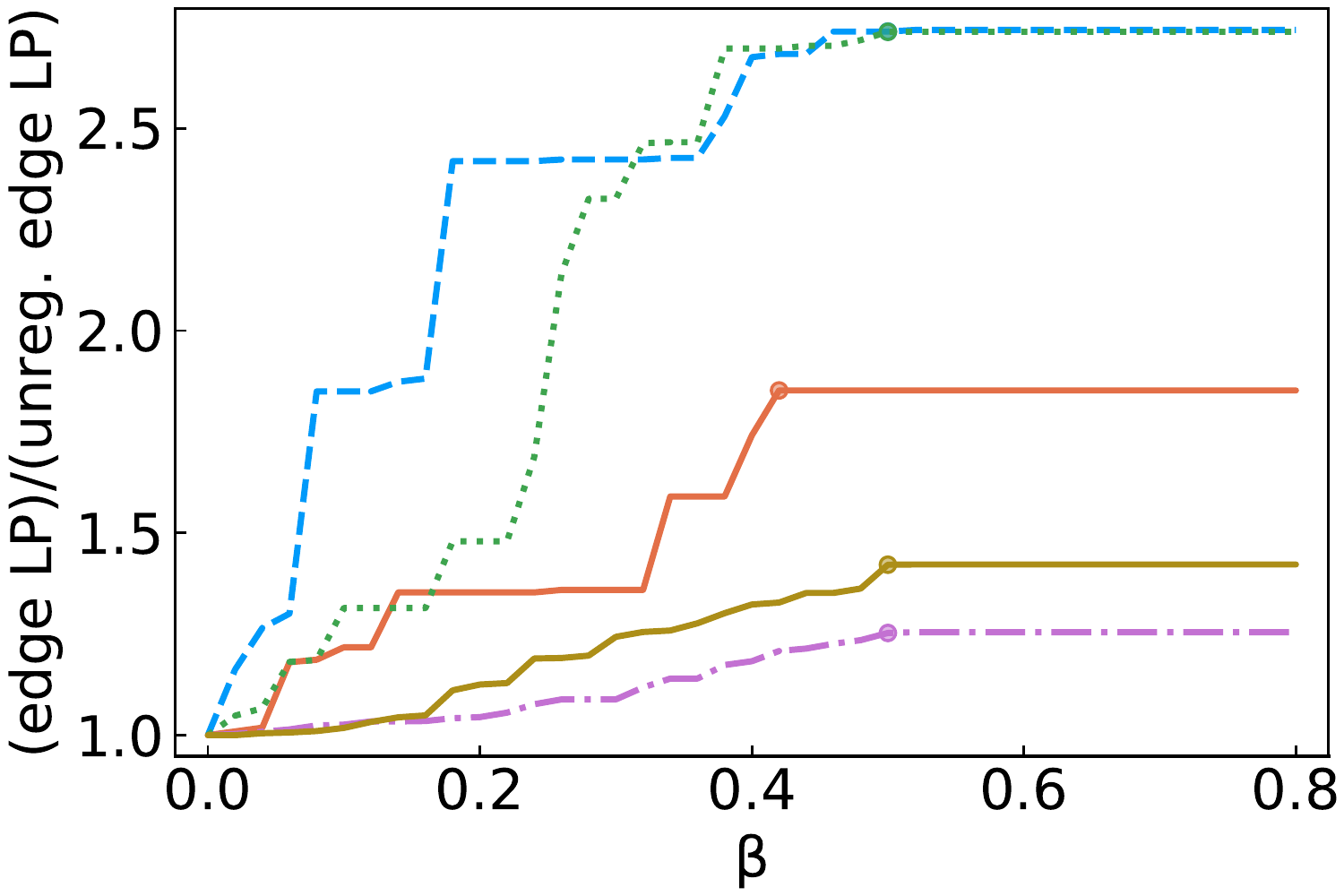}
    \includegraphics[width=0.75\columnwidth]{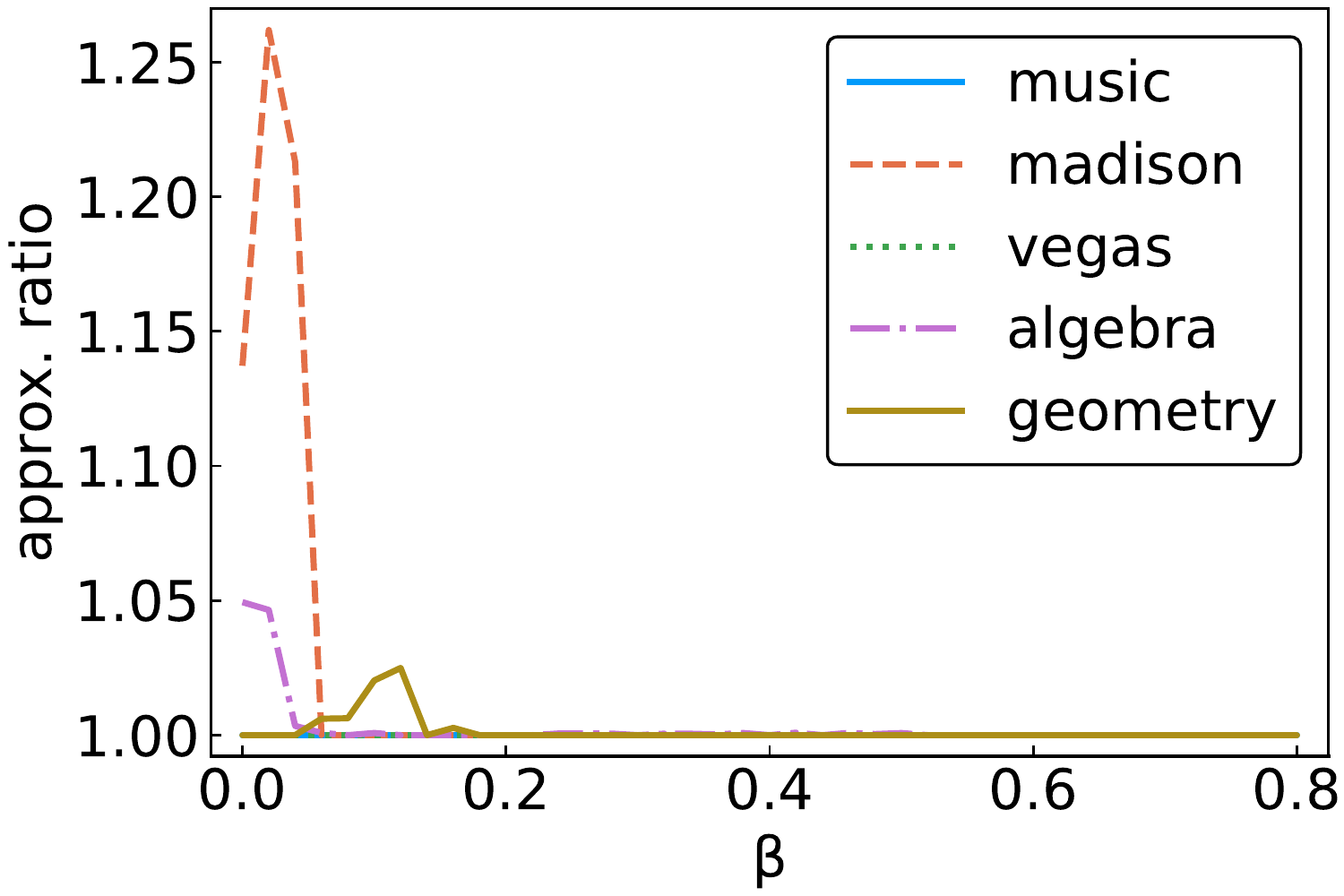}\hspace{1.5cm}
    \includegraphics[width=0.75\columnwidth]{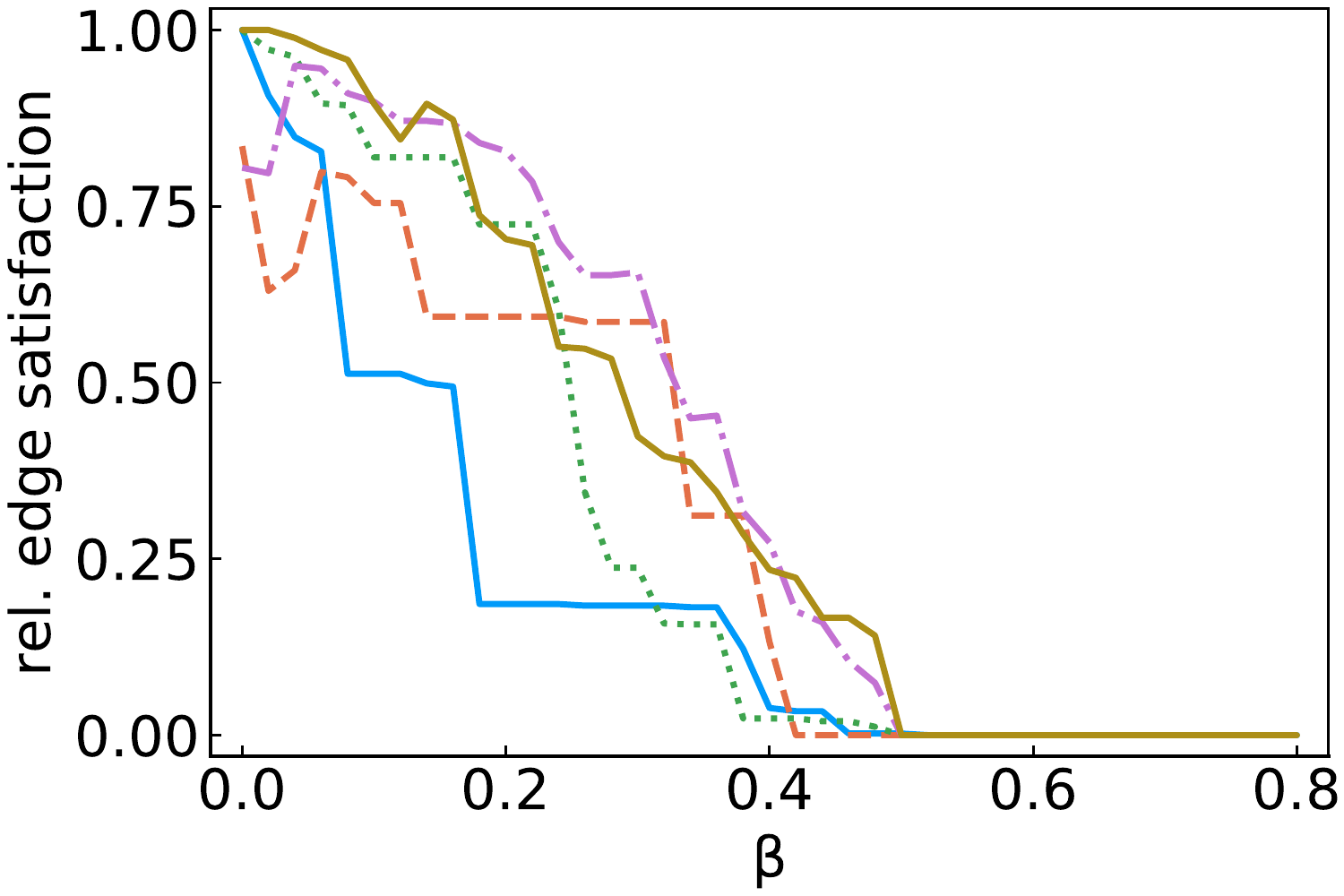}
    \caption{%
      Various performance metrics as a function of the regularization level $\beta$.
      (Top left) Ratio of Algorithm~1 cost to the cost of the optimal ILP solution with no regularization ($\beta = 0$).
      Dots mark the corresponding $\hat{\beta}$.
      (Top right) Ratio of the un-rounded relaxed LP solution of the regularized objective to the un-rounded relaxed LP solution without regularization in terms of the original edge-based clustering objective~\eqref{obj:cec} from~\cite{amburg2020clustering}.
      Dots mark the same $\hat{\beta}$.
      (Bottom left) Approximation factor of Algorithm~1 obtained by solving
      the ILP exactly.
      (Bottom right) 
      Fraction of hyperedges where the clustering assigns all nodes in the hyperedge
      to the color of the hyperedge (``edge satisfaction''), normalized by the solution at $\beta = 0$.
      }
    \label{fig:statics}
\end{figure*}

Within our framework, a diverse clustering of users from a review platform
corresponds to composing groups of users for a particular category that
contains both experts (with reviews in the given category) and those
with diverse perspectives (having reviewed other categories). 
The reviews from these users could then be used to present a ``group of reviews''
for a particular category.
A diverse clustering for the question-and-answer platforms joins users
with expertise in one particular math topic with those who have experiences in another topic.
This serves as an approximation to how one might construct experienced and diverse teams,
given historical data on individuals' experiences.


\subsection{Algorithm performance with varying regularization}

Here, we examine the performance of Algorithm 1 for various regularization
strengths $\beta$ and compare the results to the unregularized case
(Figure~\ref{fig:statics}). We observe that including the regularization term and using Algorithm 1 to solve the resulting objective
only yields mild increases in cost compared to the optimal solution of the original unregularized objective. 
This ``cost of diversity'' ratio is always smaller than 3 and is especially small for the
Math Overflow questions datasets (Figure~\ref{fig:statics}, top left). Furthermore, the ratio between the LP relaxation of the regularized objective and the LP relaxation of the unregularized ($\beta = 0$) objective has similar properties (Figure~\ref{fig:statics}, top right). This is not surprising, given that every node in each of the datasets has a color degree of zero for some color, and thus for very large values of $\beta$, each node is put in a cluster where it has a zero color degree, so that the second term in the objective is zero. Also, the approximation factor of Algorithm 1 on the data is small (Figure~\ref{fig:statics}, bottom left), which we obtain by solving the exact ILP, indicating that the relaxed LP performs very well. In fact, solving the relaxed LP often yields an integral
solution, meaning that it is the solution to the ILP. The computed
$\hat{\beta}$ bound also matches the plateau of the rounded solution
(Figure~\ref{fig:statics}, top left), which we would also expect from the small
approximation factors and the fact that each node has at least one color degree of zero. We also examine the ``edge satisfaction'', i.e., the
fraction of hyperedges where all of the nodes in the hyperedge are clustered to
the same color as the hyperedge~\cite{amburg2020clustering}
(Figure~\ref{fig:statics}, bottom right). As regularization increases, more
diversity is encouraged, and the edge satisfaction decreases. Finally, we
note that the runtime of Algorithm 1 is small in practice, taking at most a
couple of seconds in any problem instance.


\begin{figure}[tb]
   \includegraphics[width=0.80\columnwidth]{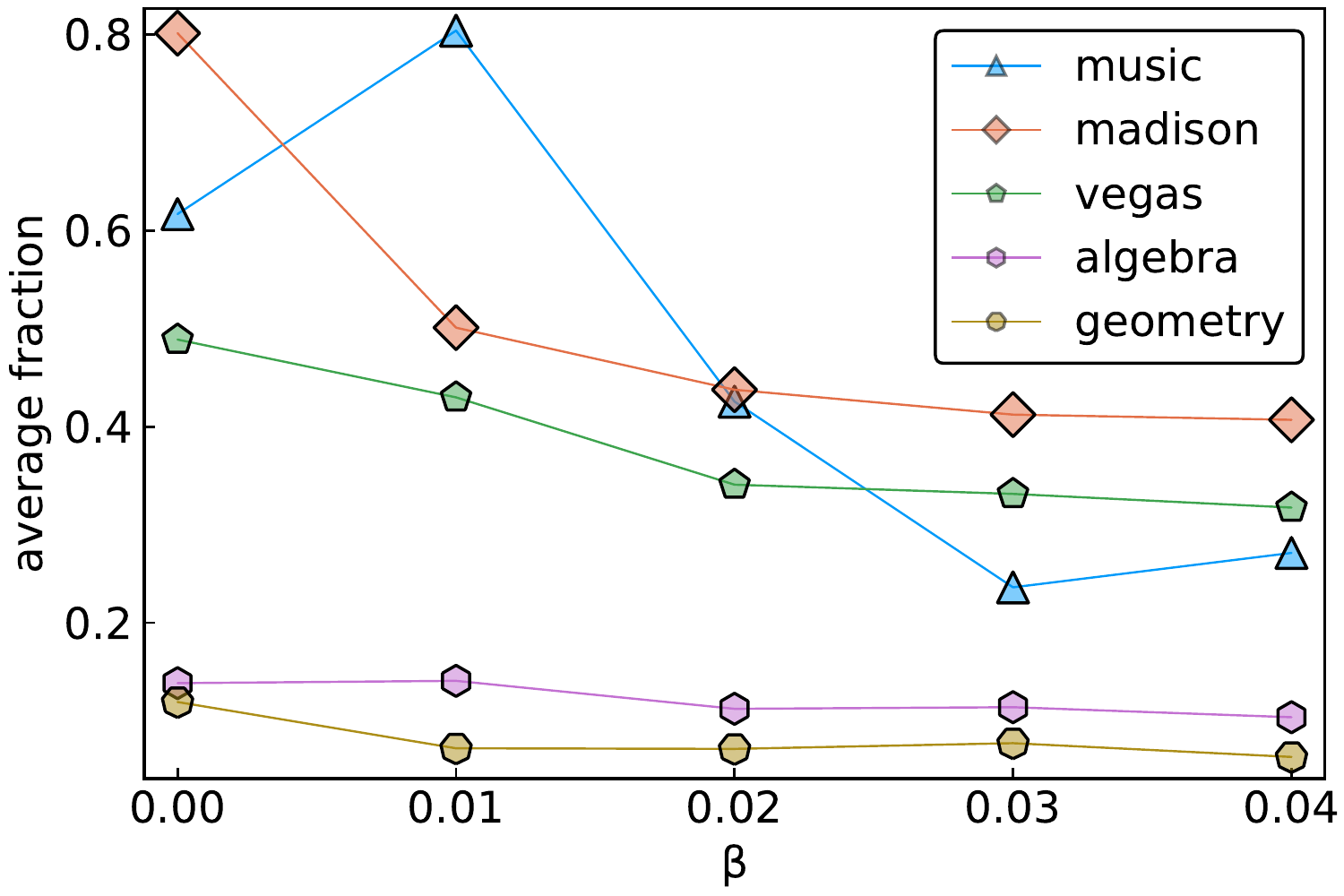}
  \caption{Average fraction ($f_{\text{within}}$) of within-cluster reviews/posts.
  As the regularization increases, bias toward same-cluster assignments decreases and the past review/post distribution within clusters becomes more diverse.}
  \label{fig:percent}
\end{figure}

Next, we examine the effect of regularization strength on diversity within clusters. 
To this end, we measure the average fraction of within-cluster reviews/posts. Formally, for a clustering of nodes $\mathcal{C},$ this measure, which we call $f_{\text{within}}$, is calculated as follows:
$$f_{\text{within}}=\sum_{i\in L}|\mathcal{C}(i)|/|V|\sum_{v\in\mathcal{C}(i)}d_v^i/d_v.$$
In computing the $f_{\text{within}}$ measure, within each cluster we compute the fraction of all user reviews/posts having the same category as the cluster. 
Then we average these fractions across all clusters, weighted by cluster size. 
Figure~\ref{fig:percent} shows that $f_{\text{within}}$ decreases with regularization strength,
indicating that our clustering framework yields meaningfully diverse clusters.

\xhdr{Case Study I: Mexican restaurants in Madison, WI}
\begin{figure}[tb]
  \includegraphics[width=0.80\columnwidth]{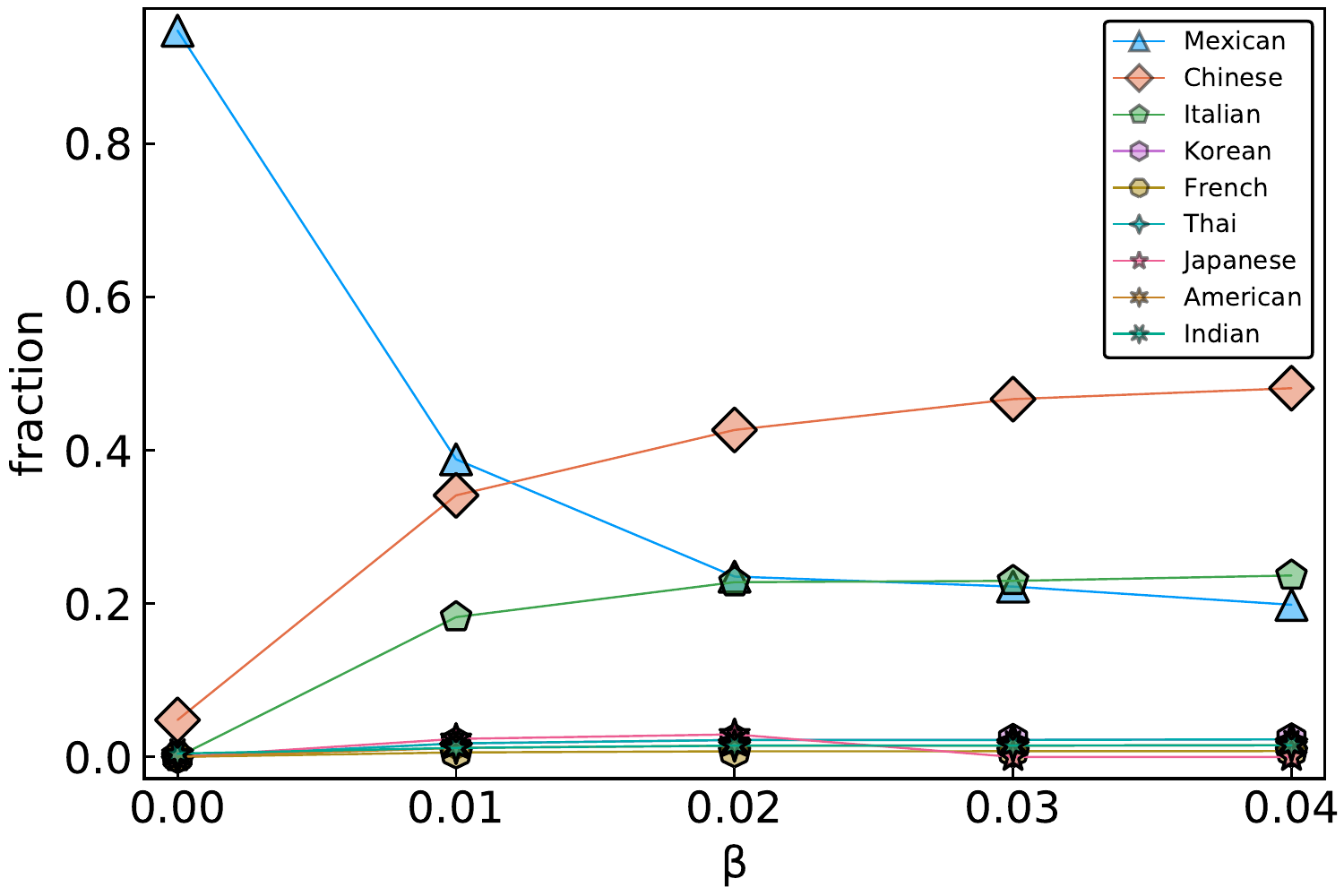}
  \caption{Distribution of node (reviewer) majority categories within the Mexican restaurant review cluster. Increased regularization leads to more experts from other cuisines being assigned to write Mexican reviews.}
  \label{fig:mexican}
\end{figure}
\begin{figure}[tb]
  \includegraphics[width=0.80\columnwidth]{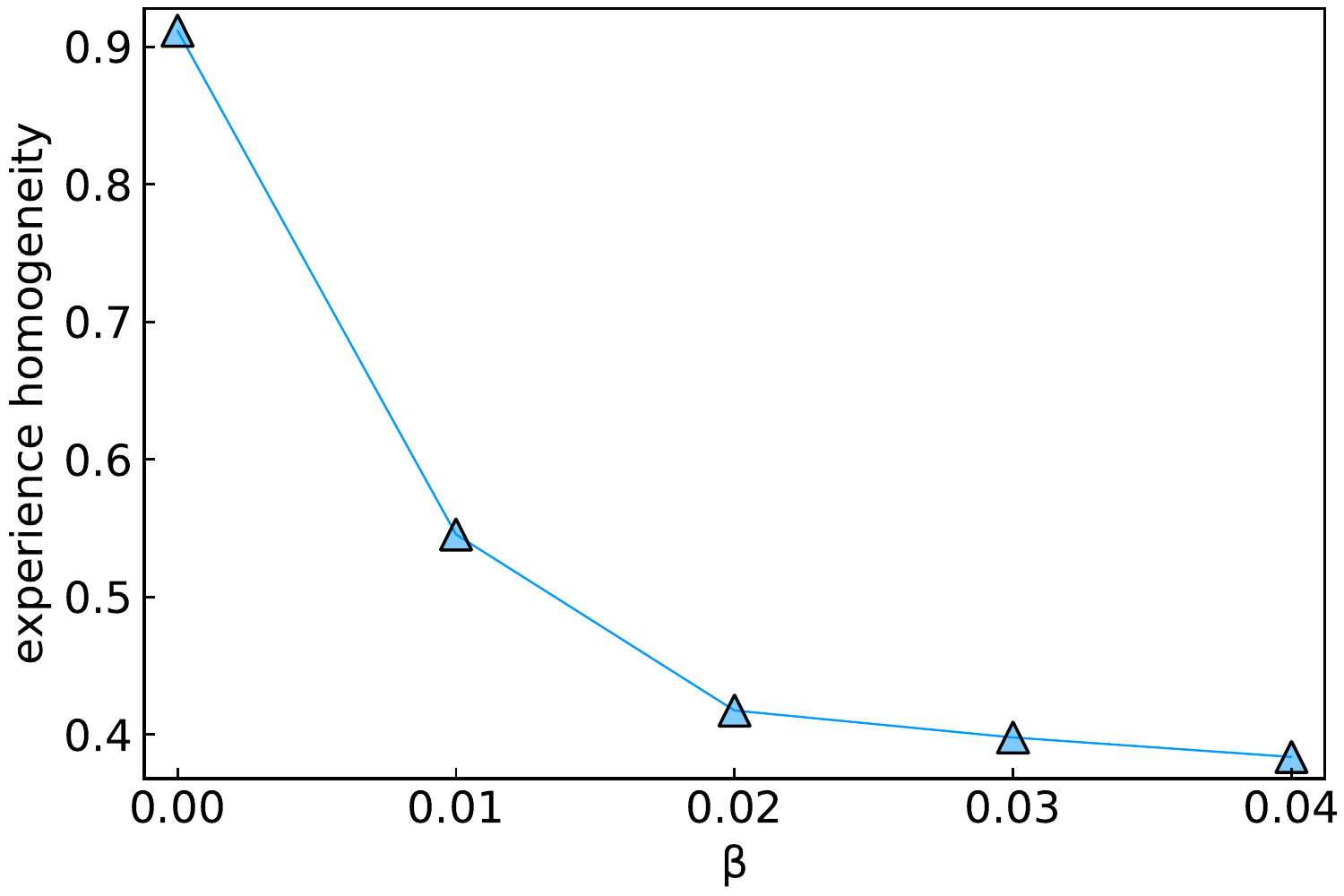}
  \caption{Among users assigned to the Mexican cluster, we compute the fraction (experience homogeneity score) of user reviews that were written in that same category. Increased regularization leads to less total past Mexican reviews among the reviewers assigned to the cluster.}
  \label{fig:percentmexican}
\end{figure}
We now take a closer look at the output of Algorithm 1 on one dataset to better understand the way in which it encourages diversity within clusters.  Taking the entire madison-restaurant-reviews dataset as a starting point, we cluster each reviewer to write reviews of restaurants falling into one of nine cuisine categories. After that, we examine the set of reviewers grouped by Algorithm 1 to review Mexican restaurants. To compare the diversity of experience for various regularization strengths, we plot the distribution of reviewers' \emph{majority categories} in Figure~\ref{fig:mexican}. Here, we take ``majority category'' to mean the majority vote assignment of each reviewer. In other words, the majority category is the one 
in which they have published the most reviews. 
We see that as $\beta$ increases, the cluster becomes more diverse, as the dominance of the Mexican majority category gradually subsides, and it is overtaken by the Chinese category. At $\beta=0$ (no regularization), 95\% of nodes in the Mexican reviewer cluster have a majority category of Mexican, while at $\beta=0.04,$ only 20\% still do. 
Thus, as regularization increases, we see greater diversity within the cluster, as ``expert'' reviewers from other cuisines are clustered to review Mexican restaurants. 

Similarly, as $\beta$ increases we see a decrease in 
the fraction of users' reviews that are for Mexican restaurants, when this fraction is averaged across all users assigned to the Mexican restaurant cluster
(Figure~\ref{fig:percentmexican}). 
We refer to this ratio as the {\it experience homogeneity score}, which for a cluster $\mathcal{C}(i)$ is formally written as 
$$\text{experience homogeneity score(i)}= \sum_{v\in\mathcal{C}(i)}d_v^i/d_v.$$ This measure is similar to $f_{\text{within}}$ except that we look at only one cluster.
However, this experience homogeneity score does not decrease as much as the corresponding fraction in Figure~\ref{fig:mexican}, falling from 91\% to 38\%, which illustrates that while the ``new'' reviewers added to the cluster with increasing $\beta$ have expertise in other areas, they nonetheless have written some reviews of Mexican restaurants in the past. 

\xhdr{Case Study II: New York blues on Amazon}
\begin{figure}[tb]
  \vspace{0.00000000cm}
  \includegraphics[width=0.80\columnwidth]{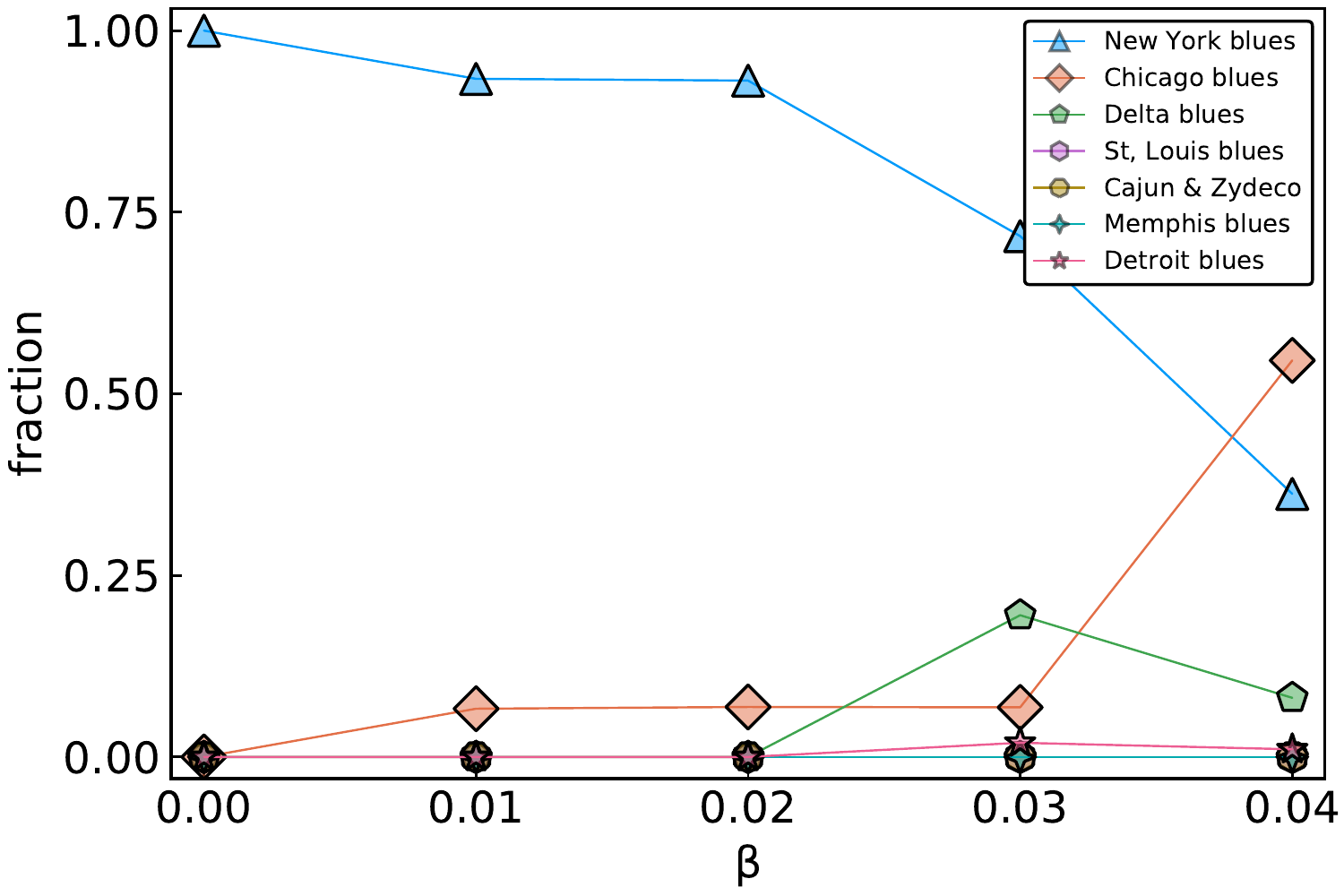}
  \caption{Distribution of node (reviewer) majority categories within the New York blues review cluster. Increased regularization leads to more experts from other genres being assigned to write New York blues reviews.}
  \label{fig:blue}
\end{figure}
\begin{figure}[tb]
 \vspace{0.21cm}
 \includegraphics[width=0.81\columnwidth]{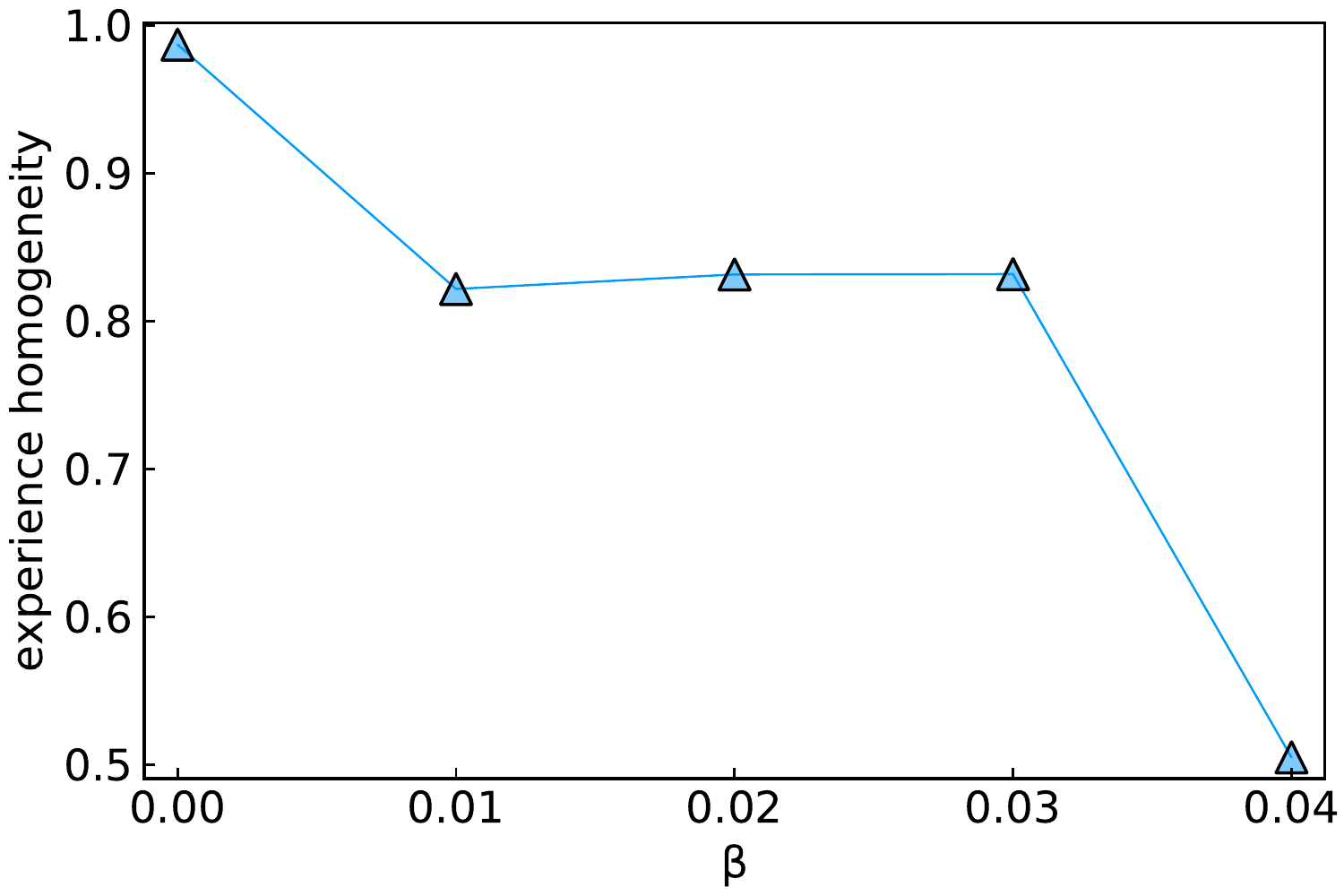}
  \caption{Among users assigned to the New York blues cluster, we compute the fraction (experience homogeneity score) of user reviews that were written in that same category.  Increased regularization leads to less total historic New York blues reviews among the reviewers assigned to the cluster.}
  \label{fig:percentblue}
\end{figure}
We now perform an analogous analysis for reviews of blues vinyls on Amazon. We take the entire music-blues-reviews dataset and perform regularized clustering at different levels of regularization. Then, we keep track of the reviewers clustered to write reviews in the New York blues category. The fraction of users in the cluster whose majority category is New York decreases as we add more regularization, as seen in Figure~\ref{fig:blue}. The percent share of all past New York blues reviews of reviewers within the cluster decreases with regularization, but not as much as the corresponding curve in the majority category distribution, as seen in Figure~\ref{fig:percentblue}, indicating that while ``new'' reviewers added to the cluster with increasing $\beta$ have expertise in other types of blues music they nonetheless have written some reviews of New York blues vinyls in the past.
\subsection{Dynamic group formation}

\begin{figure}[t]
    \centering
    \includegraphics[width=0.94\columnwidth]{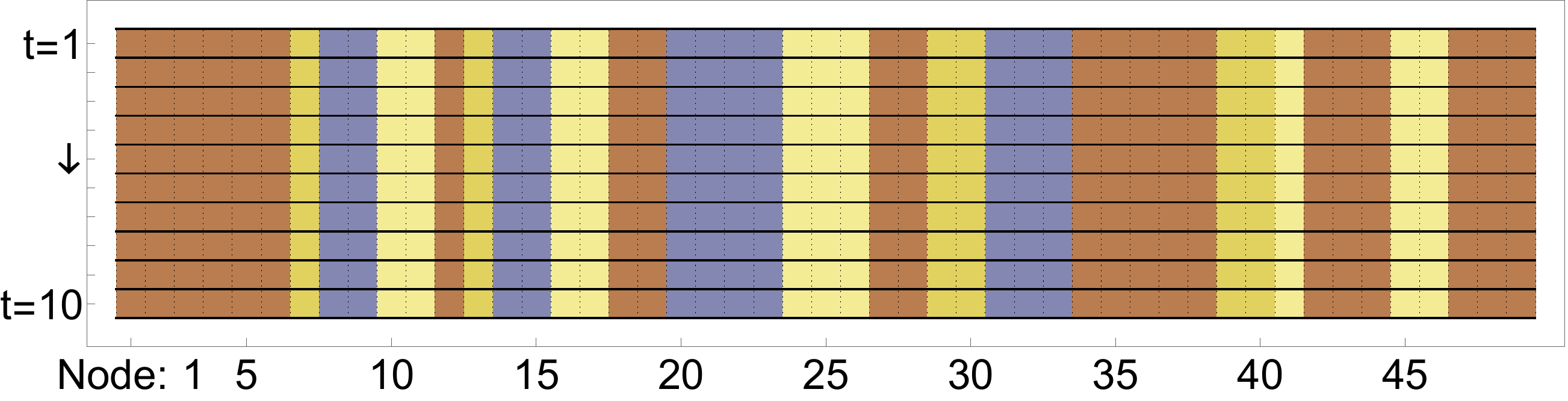}
    \includegraphics[width=0.94\columnwidth]{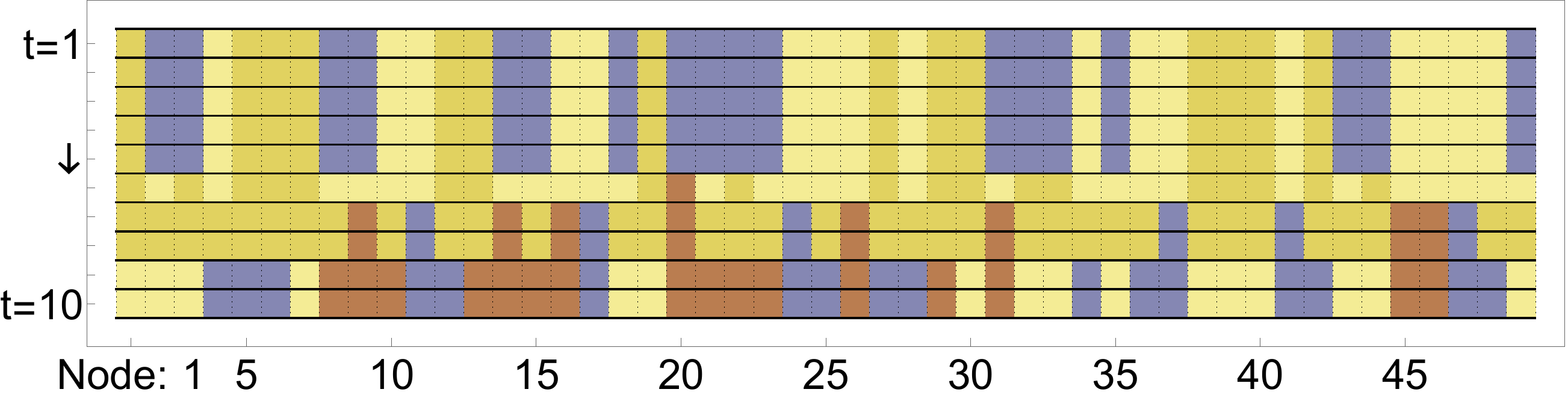}
    \includegraphics[width=0.94\columnwidth]{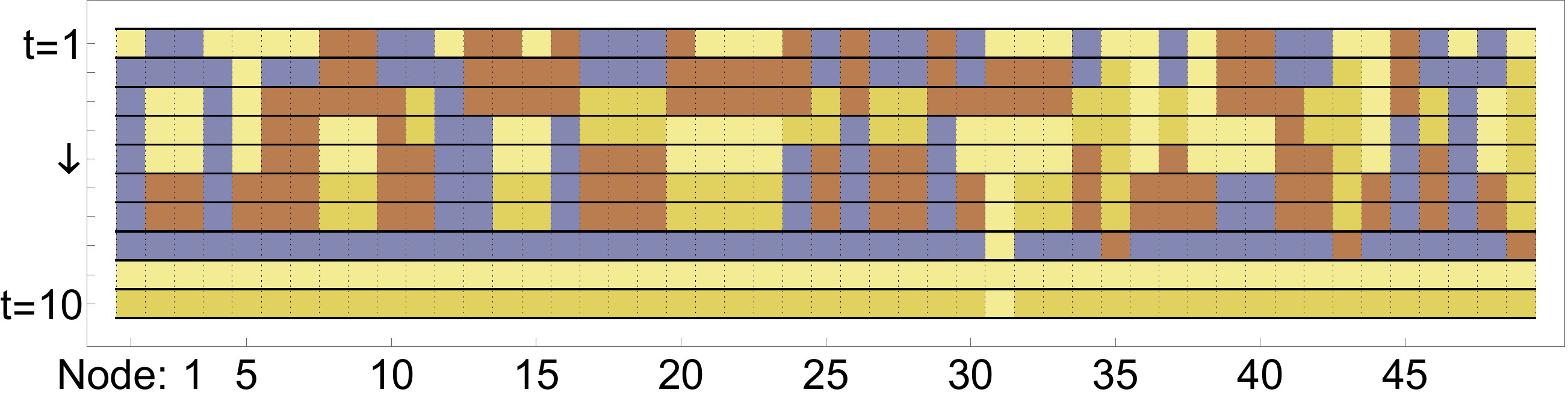}
    \includegraphics[width=0.94\columnwidth]{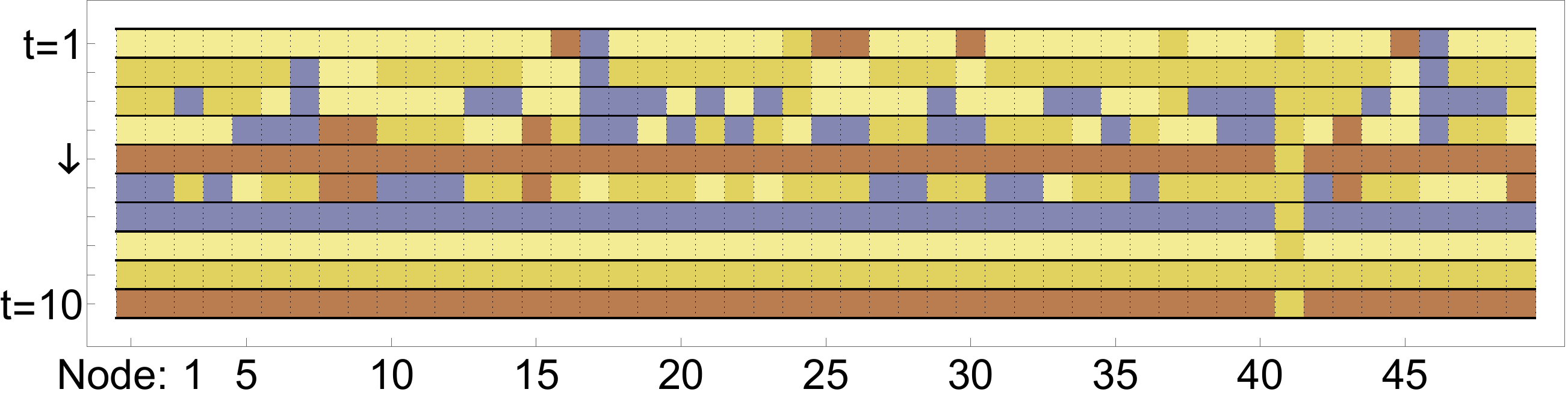}    
    \includegraphics[width=0.94\columnwidth]{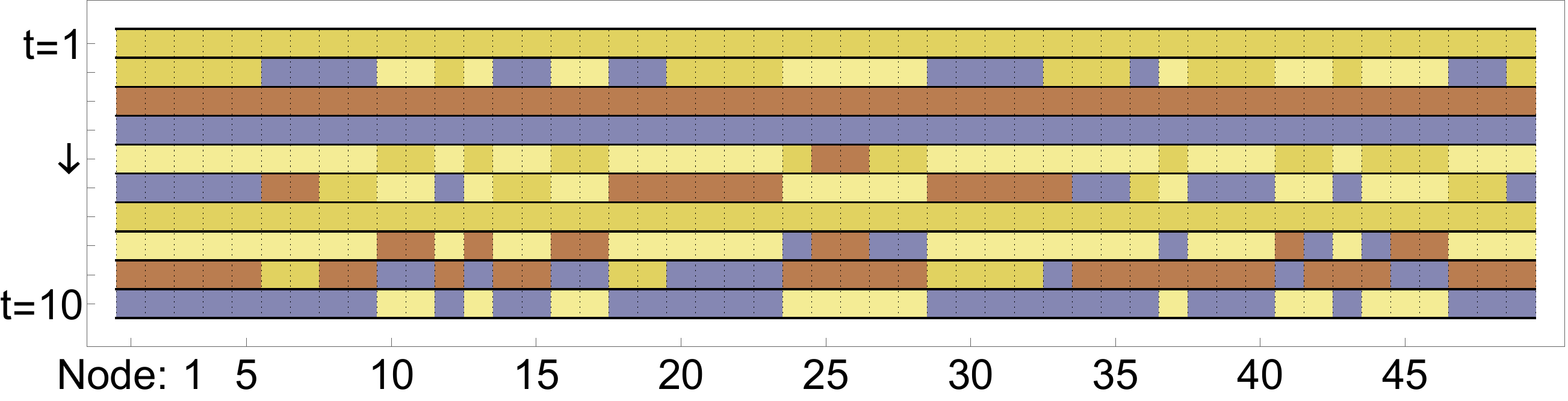}
    \includegraphics[width=0.94\columnwidth]{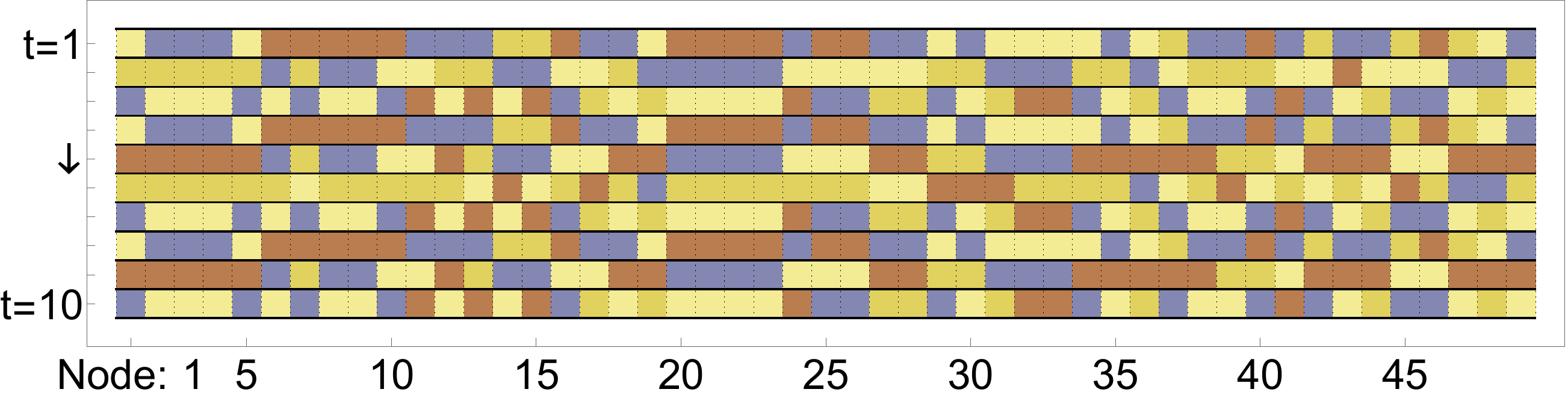}
    \caption{Color assignments over time for a subset of nodes and tags in the geometry-questions dataset
      for different regularization parameters $\beta$ (from top to bottom: $\beta =$ 0, 0.07, 0.1, 0.2, 0.4, 0.7).
      In the $\beta = 0$ case nodes are always assigned to the same color. For large enough $\beta$, nodes
      exchange at every time step (bottom).}
    \label{fig:math_colors}
\end{figure}

In this section we consider a dynamic variant of our framework in which we
iteratively update the hypergraph. More specifically, given the hypergraph up to
time $t$, we (i) solve our regularized objective to find a clustering
$\mathcal{C}$ (ii) create a set of hyperedges at time $t + 1$ corresponding to
$\mathcal{C}$, i.e., all nodes of a given color create a hyperedge. At the next
step, the experience levels of all nodes have changed.  This mimics a scenario
in which teams are repeatedly formed via Algorithm 1 for various tasks, where
the type of the task corresponds to the color of the hyperedge.  For our
experiments, we only track the experiences from a window of the last $w$ time
steps; in other words, the hypergraph just consists of the hyperedges appearing
in the previous $w$ steps.  To get an initial history for the nodes, we start
with the hypergraph datasets used above. After, we run the iteration for $w$
steps to ``warm start'' the dynamical process, and consider this state to be the
initial condition. Finally, we run the iterative procedure for $T$ times.

We can immediately analyze some limiting behavior of this process. When
$\beta=0$ (i.e., no regularization), after the first step, the clustering will
create new hyperedges that increase the experience levels of each node for some
color.  In the next step, no node has any incentive to cluster with a different
color than the previous time step, so the clustering will be the same. Thus, the
dynamical process is entirely static.  At the other extreme, if $\beta>d_{max}$
at every step of the dynamical process, then the optimal solution at each step
is a minority vote assignment by Theorem~\ref{thm:betalargelp}.  In this case,
after each step, each node $v$ will increase its color degree in one color,
which may change its minority vote solution in the next iteration.  Assuming
ties are broken at random, this leads to uniformity in the historical cluster
assignments of each node as $T \to \infty$.

\begin{figure}[t]
  \includegraphics[width=0.90\columnwidth]{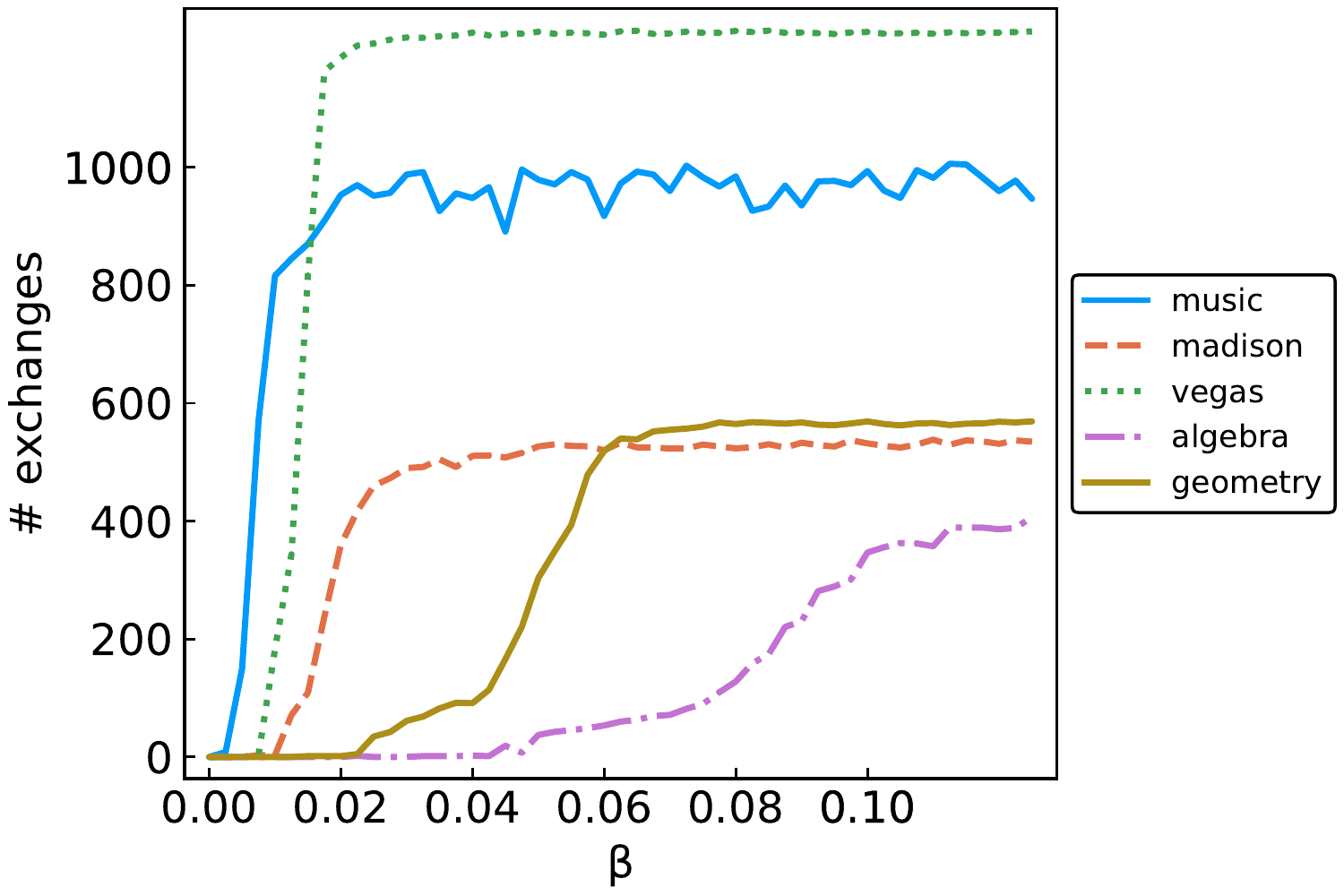}
  \caption{Mean number of node exchanges
    as a function of the diversity regularization. While nodes perpetually gain expertise at low enough $\beta$ (no exchanges), less and less experience expertise gain is required at higher regularization strengths as more and more nodes exchange at every time step.}
  \label{fig:exchanges}
\end{figure}

For each of our datasets, we ran the dynamical process for $T = 50$ time steps.
We say that a node \emph{exchanges} if it is clustered to different colors in
consecutive time steps. Figure~\ref{fig:exchanges} shows the mean number of
exchanges. As expected, for small enough $\beta$, nodes are always assigned the
same color, resulting in no exchanges; for large enough $\beta$, nearly all
nodes exchange in the minority vote regime.
Figure~\ref{fig:math_colors} shows the clustering of nodes on a subset of the geometry-questions
dataset for different regularization levels. For small $\beta$, nodes accumulate
experience before exchanging. When $\beta$ is large, nodes exchange at every iteration.
This corresponds to the scenario of large $\beta$ in Figure~\ref{fig:exchanges}.


\section{Discussion}

We present a new framework for clustering that balances diversity and experience in cluster formation. 
We cast our problem as a hypergraph clustering task, where a regularization parameter controls cluster diversity, and
we have an algorithm that achieves a 2-approximation for any value of the regularization parameter. 
In numerical experiments, the approximation algorithm is effective and finds solutions that are nearly as good
as the unregularized objective. In fact, the linear programming relaxation often returns an integral solution, so
the algorithm finds the optimal solution in many cases.

Managing hyperparameters is generally a nuanced task, particularly for
unsupervised problems such as ours.  Remarkably, within our framework, we were
able to characterize solutions for extremal values of the regularization
parameter and also numerically determine intervals for which the regularization
parameter provides a meaningful tradeoff for our objective.  This type of
technique relied on the linear programming formulation of the problem and could
be of interest to a broader set of problems.  Interestingly, as the
regularization parameter changes from zero to infinity, our problem transitions
from being NP-hard to polynomial time solvable.  In future work, we plan to
explore how and when this transition occurs, and in particular whether we can
obtain better parameter-dependent approximation guarantees.


\bibliographystyle{plain}
\bibliography{main}

\begin{thebibliography}{10}

\bibitem{adler1992geometric}
Ilan Adler and Renato~D.C. Monteiro.
\newblock A geometric view of parametric linear programming.
\newblock {\em Algorithmica}, 8:161--176, 1992.

\bibitem{ahmadi2020fair}
Saba Ahmadi, Sainyam Galhotra, Barna Saha, and Roy Schwartz.
\newblock Fair correlation clustering.
\newblock {\em arXiv preprint arXiv:2002.03508}, 2020.

\bibitem{ahmadian2019clustering}
Sara Ahmadian, Alessandro Epasto, Ravi Kumar, and Mohammad Mahdian.
\newblock Clustering without over-representation.
\newblock In {\em Proceedings of the 25th ACM SIGKDD International Conference
  on Knowledge Discovery \& Data Mining}, pages 267--275, 2019.

\bibitem{ahmadian2019correlation}
Sara Ahmadian, Alessandro Epasto, Ravi Kumar, and Mohammad Mahdian.
\newblock Fair correlation clustering.
\newblock In {\em Proceedings of the 23rd International Conference on
  Artificial Intelligence and Statistics}, 2020.

\bibitem{amburg2020clustering}
Ilya Amburg, Nate Veldt, and Austin Benson.
\newblock Clustering in graphs and hypergraphs with categorical edge labels.
\newblock In {\em Proceedings of The Web Conference 2020}, pages 706--717,
  2020.

\bibitem{aziz2019rule}
Haris Aziz.
\newblock A rule for committee selection with soft diversity constraints.
\newblock {\em Group Decision and Negotiation}, 28(6):1193--1200, 2019.

\bibitem{bansal2004correlation}
Nikhil Bansal, Avrim Blum, and Shuchi Chawla.
\newblock Correlation clustering.
\newblock {\em Machine learning}, 56(1-3):89--113, 2004.

\bibitem{barocas2016big}
Solon Barocas and Andrew~D Selbst.
\newblock Big data's disparate impact.
\newblock {\em California Law Review}, 104:671, 2016.

\bibitem{benson2016higher}
Austin~R. Benson, David~F. Gleich, and Jure Leskovec.
\newblock Higher-order organization of complex networks.
\newblock {\em Science}, 353(6295):163--166, 2016.

\bibitem{bera2019fair}
Suman Bera, Deeparnab Chakrabarty, Nicolas Flores, and Maryam Negahbani.
\newblock Fair algorithms for clustering.
\newblock In {\em Advances in Neural Information Processing Systems}, pages
  4955--4966, 2019.

\bibitem{bonchi2015chromatic}
Francesco Bonchi, Aristides Gionis, Francesco Gullo, Charalampos~E Tsourakakis,
  and Antti Ukkonen.
\newblock Chromatic correlation clustering.
\newblock {\em ACM Transactions on Knowledge Discovery from Data (TKDD)},
  9(4):1--24, 2015.

\bibitem{castells2015novelty}
Pablo Castells, Neil~J Hurley, and Saul Vargas.
\newblock Novelty and diversity in recommender systems.
\newblock In {\em Recommender systems handbook}, pages 881--918. Springer,
  2015.

\bibitem{celis2018multiwinner}
L~Elisa Celis, Lingxiao Huang, and Nisheeth~K Vishnoi.
\newblock Multiwinner voting with fairness constraints.
\newblock In {\em Proceedings of the 27th International Joint Conference on
  Artificial Intelligence}, pages 144--151, 2018.

\bibitem{chen2018fast}
Laming Chen, Guoxin Zhang, and Eric Zhou.
\newblock Fast greedy map inference for determinantal point process to improve
  recommendation diversity.
\newblock In {\em Advances in Neural Information Processing Systems}, pages
  5622--5633, 2018.

\bibitem{chen2019proportionally}
Xingyu Chen, Brandon Fain, Liang Lyu, and Kamesh Munagala.
\newblock Proportionally fair clustering.
\newblock In {\em International Conference on Machine Learning}, pages
  1032--1041, 2019.

\bibitem{chierichetti2017fair}
Flavio Chierichetti, Ravi Kumar, Silvio Lattanzi, and Sergei Vassilvitskii.
\newblock Fair clustering through fairlets.
\newblock In {\em Advances in Neural Information Processing Systems}, pages
  5029--5037, 2017.

\bibitem{chouldechova2017fair}
Alexandra Chouldechova.
\newblock Fair prediction with disparate impact: A study of bias in recidivism
  prediction instruments.
\newblock {\em Big data}, 5(2):153--163, 2017.

\bibitem{corbett2018measure}
Sam Corbett-Davies and Sharad Goel.
\newblock The measure and mismeasure of fairness: A critical review of fair
  machine learning.
\newblock {\em arXiv preprint arXiv:1808.00023}, 2018.

\bibitem{corbett2017algorithmic}
Sam Corbett-Davies, Emma Pierson, Avi Feller, Sharad Goel, and Aziz Huq.
\newblock Algorithmic decision making and the cost of fairness.
\newblock In {\em Proceedings of the 23rd ACM SIGKDD International Conference
  on Knowledge Discovery and Data Mining}, pages 797--806, 2017.

\bibitem{davidson2019making}
Ian Davidson and SS~Ravi.
\newblock Making existing clusterings fairer: Algorithms, complexity results
  and insights.
\newblock Technical report, University of California, Davis., 2019.

\bibitem{forsyth2018group}
Donelson~R Forsyth.
\newblock {\em Group dynamics}.
\newblock Cengage Learning, 2018.

\bibitem{gan2020graph}
Junhao Gan, David~F. Gleich, Nate Veldt, Anthony Wirth, and Xin Zhang.
\newblock {Graph Clustering in All Parameter Regimes}.
\newblock In Javier Esparza and Daniel Kr{\'a}ľ, editors, {\em 45th
  International Symposium on Mathematical Foundations of Computer Science (MFCS
  2020)}, volume 170 of {\em Leibniz International Proceedings in Informatics
  (LIPIcs)}, pages 39:1--39:15, Dagstuhl, Germany, 2020. Schloss
  Dagstuhl--Leibniz-Zentrum f{\"u}r Informatik.

\bibitem{hadley1995}
Scott~W. Hadley.
\newblock Approximation techniques for hypergraph partitioning problems.
\newblock {\em Discrete Applied Mathematics}, 59(2):115--127, 1995.

\bibitem{homan2008facing}
Astrid~C Homan, John~R Hollenbeck, Stephen~E Humphrey, Daan~Van Knippenberg,
  Daniel~R Ilgen, and Gerben~A Van~Kleef.
\newblock Facing differences with an open mind: Openness to experience,
  salience of intragroup differences, and performance of diverse work groups.
\newblock {\em Academy of Management Journal}, 51(6):1204--1222, 2008.

\bibitem{ihler1993modeling}
Edmund Ihler, Dorothea Wagner, and Frank Wagner.
\newblock Modeling hypergraphs by graphs with the same mincut properties.
\newblock {\em Information Processing Letters}, 45(4):171--175, 1993.

\bibitem{jackson1995diversity}
Susan~E Jackson and Marian~N Ruderman.
\newblock {\em Diversity in work teams: Research paradigms for a changing
  workplace.}
\newblock American Psychological Association, 1995.

\bibitem{yelp}
Kaggle.
\newblock Yelp dataset.
\newblock \url{https://www.kaggle.com/yelp-dataset/yelp-dataset}, 2020.

\bibitem{kleinberg2018algorithmic}
Jon Kleinberg, Jens Ludwig, Sendhil Mullainathan, and Ashesh Rambachan.
\newblock Algorithmic fairness.
\newblock In {\em {AEA} papers and proceedings}, volume 108, pages 22--27,
  2018.

\bibitem{kleinberg2018team}
Jon Kleinberg and Maithra Raghu.
\newblock Team performance with test scores.
\newblock {\em ACM Transactions on Economics and Computation (TEAC)},
  6(3-4):1--26, 2018.

\bibitem{kleindessner2019guarantees}
Matth{\"a}us Kleindessner, Samira Samadi, Pranjal Awasthi, and Jamie
  Morgenstern.
\newblock Guarantees for spectral clustering with fairness constraints.
\newblock In {\em International Conference on Machine Learning}, pages
  3458--3467, 2019.

\bibitem{kunaver2017diversity}
Matev{\v{z}} Kunaver and Toma{\v{z}} Po{\v{z}}rl.
\newblock Diversity in recommender systems--a survey.
\newblock {\em Knowledge-Based Systems}, 123:154--162, 2017.

\bibitem{lawler1973cutsets}
Eugene~L Lawler.
\newblock Cutsets and partitions of hypergraphs.
\newblock {\em Networks}, 3(3):275--285, 1973.

\bibitem{levi2015group}
Daniel Levi.
\newblock {\em Group dynamics for teams}.
\newblock Sage Publications, 2015.

\bibitem{li2017inhomogeneous}
Pan Li and Olgica Milenkovic.
\newblock Inhomogeneous hypergraph clustering with applications.
\newblock In {\em Advances in Neural Information Processing Systems}, pages
  2308--2318, 2017.

\bibitem{lu2012recommender}
Linyuan L{\"u}, Mat{\'u}{\v{s}} Medo, Chi~Ho Yeung, Yi-Cheng Zhang, Zi-Ke
  Zhang, and Tao Zhou.
\newblock Recommender systems.
\newblock {\em Physics reports}, 519(1):1--49, 2012.

\bibitem{machado2019fair}
Lucas Machado and Kostas Stefanidis.
\newblock Fair team recommendations for multidisciplinary projects.
\newblock In {\em 2019 IEEE/WIC/ACM International Conference on Web
  Intelligence (WI)}, pages 293--297. IEEE, 2019.

\bibitem{mehrabi2019survey}
Ninareh Mehrabi, Fred Morstatter, Nripsuta Saxena, Kristina Lerman, and Aram
  Galstyan.
\newblock A survey on bias and fairness in machine learning.
\newblock {\em arXiv preprint arXiv:1908.09635}, 2019.

\bibitem{ni2019justifying}
Jianmo Ni, Jiacheng Li, and Julian McAuley.
\newblock Justifying recommendations using distantly-labeled reviews and
  fine-grained aspects.
\newblock In {\em Proceedings of the 2019 Conference on Empirical Methods in
  Natural Language Processing and the 9th International Joint Conference on
  Natural Language Processing (EMNLP-IJCNLP)}, pages 188--197, 2019.

\bibitem{nowozin2009solution}
Sebastian Nowozin and Stefanie Jegelka.
\newblock Solution stability in linear programming relaxations: Graph
  partitioning and unsupervised learning.
\newblock In {\em Proceedings of the 26th Annual International Conference on
  Machine Learning}, pages 769--776. ACM, 2009.

\bibitem{santos2010exploiting}
Rodrygo~LT Santos, Craig Macdonald, and Iadh Ounis.
\newblock Exploiting query reformulations for web search result
  diversification.
\newblock In {\em Proceedings of the 19th international conference on World
  wide web}, pages 881--890, 2010.

\bibitem{stratigi2020fair}
Maria Stratigi, Jyrki Nummenmaa, Evaggelia Pitoura, and Kostas Stefanidis.
\newblock Fair sequential group recommendations.
\newblock In {\em Proceedings of the 35th ACM/SIGAPP Symposium on Applied
  Computing, SAC}, 2020.

\bibitem{tsourakakis2017scalable}
Charalampos~E Tsourakakis, Jakub Pachocki, and Michael Mitzenmacher.
\newblock Scalable motif-aware graph clustering.
\newblock In {\em Proceedings of the 26th International Conference on World
  Wide Web}, pages 1451--1460, 2017.

\bibitem{vargas2011rank}
Sa{\'u}l Vargas and Pablo Castells.
\newblock Rank and relevance in novelty and diversity metrics for recommender
  systems.
\newblock In {\em Proceedings of the fifth ACM conference on Recommender
  systems}, pages 109--116, 2011.

\bibitem{veldt2020hypergraph}
Nate Veldt, Austin~R Benson, and Jon Kleinberg.
\newblock Hypergraph cuts with general splitting functions.
\newblock {\em arXiv preprint arXiv:2001.02817}, 2020.

\bibitem{zehlike2017fa}
Meike Zehlike, Francesco Bonchi, Carlos Castillo, Sara Hajian, Mohamed Megahed,
  and Ricardo Baeza-Yates.
\newblock {FA$^*$IR}: A fair top-k ranking algorithm.
\newblock In {\em Proceedings of the 2017 ACM on Conference on Information and
  Knowledge Management}, pages 1569--1578, 2017.

\bibitem{zhou2010solving}
Tao Zhou, Zolt{\'a}n Kuscsik, Jian-Guo Liu, Mat{\'u}{\v{s}} Medo,
  Joseph~Rushton Wakeling, and Yi-Cheng Zhang.
\newblock Solving the apparent diversity-accuracy dilemma of recommender
  systems.
\newblock {\em Proceedings of the National Academy of Sciences},
  107(10):4511--4515, 2010.

\end{thebibliography}

\end{document}